\documentclass[journal]{IEEEtran}%

\usepackage[utf8]{inputenc}
\usepackage{microtype}
\usepackage{amsmath}
\usepackage{amssymb}
\usepackage{amsthm}
\usepackage{mathtools}
\usepackage{etoolbox}
\usepackage{url}
\usepackage{color}
\usepackage[table,svgnames]{xcolor}
\usepackage[colorlinks=true,linkcolor=black,citecolor=MidnightBlue,urlcolor=MidnightBlue]{hyperref}
\usepackage{pdfpages}
\usepackage{graphicx}
\usepackage[capitalize,nameinlink]{cleveref}
\usepackage{xspace}
\usepackage{enumitem}
\usepackage{comment}

\usepackage{thmtools}
\usepackage{thm-restate}

\usepackage{algorithm}
\usepackage{algorithmicx}
\usepackage{algpseudocode}
\usepackage{booktabs}
\usepackage{marvosym}

\newtheorem{theorem}{Theorem}[section] %
\newtheorem{lemma}[theorem]{Lemma}
\newtheorem{corollary}[theorem]{Corollary}
\newtheorem{observation}[theorem]{Observation}
\newtheorem{proposition}[theorem]{Proposition}
\newtheorem{definition}[theorem]{Definition}

\newtheorem{claim}[theorem]{Claim}

\usepackage{cite}
\usepackage{fixltx2e}

\theoremstyle{definition} %
\newtheorem{remark}{Remark}
\AtEndEnvironment{remark}{\null\hfill\HexaSteel}%

\newcommand{\changelocaltocdepth}[1]{%
  \addtocontents{toc}{\protect\setcounter{tocdepth}{#1}}%
  \setcounter{tocdepth}{#1}%
}

\newcommand{\F}{\mathbb{F}}
\newcommand{\N}{\mathbb{N}}

\newcommand{\EHC}{\mathsf{EHC}}
\newcommand{\eps}{\varepsilon}
\newcommand{\ch}{\mathsf{Ch}}

\newcommand{\eqdef}{\stackrel{\rm def}{=}}
\newcommand{\poly}{\mathrm{poly}}

\newcommand{\Err}{\mathsf{Err}}
\newcommand{\CC}{\mathsf{CC}}
\newcommand{\RC}{\mathsf{RC}}
\newcommand{\K}{K}
\newcommand{\Unif}{\mathcal{U}}

\newcommand{\var}{var}
\newcommand{\D}{\mathsf{D}^*}
\newcommand{\wt}{\widetilde}
\newcommand{\wtD}{\wt{\D}}

\newcommand{\depth}{\mathsf{depth}}

\newcommand{\Alg}{Algorithm~{1}\xspace}
\newcommand{\AlgCRS}{Algorithm~{A}\xspace}  %
\newcommand{\AlgNO}{Algorithm~{B}\xspace}    %

\newcommand{\gnote}[1]{}
\newcommand{\ran}[1]{}
\newcommand\Tstrut{\rule{0pt}{2.3ex}}

\begin{document}

\title{Efficient Multiparty Interactive Coding, Part~I:  Oblivious Insertions, Deletions and Substitutions}

\author{Ran Gelles, Yael T. Kalai, and Govind Ramnarayan%
\thanks{A preliminary version of this work~\cite{GKR19} appeared in the Proceedings of the 2019 ACM Symposium on Principles of Distributed Computing (PODC~2019).}%
\thanks{R.~Gelles is with the Faculty of Engineering, Bar-Ilan University, Israel (e-mail: ran.gelles@biu.ac.il). Work supported in part by the Israel Science Foundation (ISF) through grant No.\@ 1078/17.}%
\thanks{Y.~T.~Kalai is with Microsoft Research and MIT (e-mail: yael@microsoft.com).}
\thanks{G.~Ramnarayan was with the EECS Department at MIT, and currently works for Neural Magic (e-mail: govind.ramnarayan@gmail.com). Work supported in part by Army Research Office grant W911NF1910217,  NSF award CCF 1665252 and NSF award DMS-1737944.}%
}

\maketitle

\begin{abstract}
In the field of interactive coding,
two or more parties wish to carry out a distributed computation
over a communication network that may be noisy.
The ultimate goal is to develop efficient coding schemes that can tolerate a high
level of noise while increasing the communication by only a constant factor (i.e., constant rate).

In this work we consider synchronous communication networks over an arbitrary topology, in the powerful adversarial insertion-deletion noise model.  Namely, the noisy channel may adversarially alter the content of any transmitted symbol, as well as completely remove a transmitted symbol or inject a new symbol into the channel.

We provide an efficient, 
constant rate scheme that conducts any computation on any arbitrary network,
and succeeds with high probability
as long as an oblivious adversary corrupts at most $\frac{\eps}{m}$ fraction of the total communication,
where $m$ is the number of links in the network and $\eps$ is a small constant.
In this work (the first part), 
our scheme assumes that the parties share a random string to which the adversarial noise is oblivious.

While previous work considered the insertion-deletion noise model in the two-party setting,
to the best of our knowledge, our scheme is the first multiparty scheme that is resilient to insertions and deletions.
Furthermore, our scheme is the first \emph{computationally efficient} scheme in the multiparty setting that is resilient to \emph{adversarial} noise.

\end{abstract}


\setcounter{page}{1}	%

\section{Introduction}
\IEEEPARstart{C}{ommunication} channels may introduce noise of different types, e.g., flipping transmitted bits. 
One notorious type of
noise is \emph{insertion and deletion} noise, that may add or remove bits from the transmissions due to synchronization mismatch~\cite{sellers62}. 
The seminal work of Levenshtein~\cite{levenshtein1966bcc} was the first to consider codes that correct insertions and deletions leading to a long line of research on correcting such errors, and bounding the capabilities of codes correcting such errors, e.g.,
\cite{Ullman67,TK76,Tenengolts84,Levenshtein92,SZ99,DM01,GL16,SWGY17}. These codes apply to sending a transmission over a uni-directional channel.

In the early 90's, Schulman~\cite{schulman92,schulman96} introduced and studied the problem of performing an \emph{interactive} two-party computation over a noisy communication
channel, rather than communicating in a uni-directional manner. 
In \cite{RS94}, Rajagopalan and Schulman extended the two-party case and considered a network of~$n$ parties that wish to compute some function of their private inputs
by communicating over an arbitrary\footnote{By ``arbitrary'' we mean that the topology of the network can be an arbitrary graph~$G=(V,E)$ where each node is a party and each edge is a communication channel connecting the parties associated with these nodes.} \emph{noisy} network.
The work of~\cite{RS94} shows that if each channel
is the binary symmetric channel\footnote{That is, a channel that flips every bit with some constant probability $\eps \in (0,1/2)$.} (BSC$_\eps$), then one can obtain a \emph{coding scheme}
that takes any protocol $\Pi$ that assumes noiseless communication, and converts it into
a resilient protocol that computes the same task over the noisy network.

The coding scheme in~\cite{RS94} defies noise by adding redundancy.
The amount of added redundancy is usually measured with respect to the noiseless setting---the \emph{rate} of the coding is the communication of the noiseless protocol divided by the communication of the noise-resilient one.
The rate assumes values between zero and one, and ideally is bounded away from zero, commonly known as \emph{constant} or \emph{positive} rate.
The rate may vary according to the network in consideration, for instance, the rate in~\cite{RS94}
behaves as $1/O(\log (d+1))$ where $d$ is the maximal degree in the network. Hence, for networks where the maximal degree is non-constant, the rate approaches zero as the network size increases.

The next major step for multiparty coding schemes was provided by Jain et~al.~\cite{JKL15} and by Hoza and Schulman~\cite{HS16}.  In these works the noise is no longer assumed to be stochastic but instead is adversarial. That is, they consider worst-case noise where the only limit is the number of bits flipped by the adversary. They showed that as long as the adversary flips at most $\eps/m$-fraction of the total communication, a coding scheme with a constant rate can be achieved, where $\eps$ is some small constant, and $m$ is the number of communication links in the network.\footnote{Jain et al.~\cite{JKL15} obtained this result for the specific star network, whereas \cite{HS16} generalized this result to a network with arbitrary topology.} 

While both these works consider adversarial errors, they consider different communication models.~\cite{HS16} assumes that every party sends a single bit to all its neighbors in each round; this model is sometimes called \emph{fully utilized}, and is commonly considered in multiparty interactive coding~\cite{RS94,HS16,ABEGH19,BEGH18}. 
On the other hand,~\cite{JKL15} uses a relaxed communication setting, in which parties may or may not speak in a given round; this setting is very common for distributed computations, and is relatively less studied in previous work on interactive coding. Furthermore, they do not assume that the \emph{underlying} protocol is fully utilized, and simply use the (weaker) assumption that the order of speaking in the underlying protocol is known, and is independent of the parties' inputs. Naturally, one can convert any protocol in the non-fully-utilized model to a fully-utilized protocol by forcing all parties to speak at every round, and then apply an interactive coding scheme to the fully-utilized protocol. However, the conversion to a fully-utilized protocol may cause the communication complexity to increase by a factor of up to~$m$, greatly harming the rate of the coding scheme.

The above work focused on bit-flips (\emph{substitution noise}). 
In this work we consider the stronger type of noise of \emph{insertions and deletions} %
where the noise may completely remove a transmission (so that the receiver is not aware that a bit was sent to him), or inject new transmissions (so that the receiver receives a bit while the sender didn't send anything). Insertion and deletion noise is more general, and is considered to be more difficult, than bit-flips. Indeed, a bit flip can be simulated by a deletion followed by an insertion. 

In asynchronous communication networks, insertions and deletions have a destructive effect: a party might wait indefinitely for a message that was deleted while assuming it hasn't arrived yet due to the asynchronous nature of the channel. Very few coding schemes were designed for asynchronous networks (see \cite{JKL15,CGH19} and related work below). In this work we focus on synchronous networks only.

Coping with insertions and deletions in \emph{synchronous} networks is highly non-trivial, even if these are less harmful than in the asynchronous case. 
Note that insertions and deletions are trivially correctable in the fully utilized communication model; indeed, each party expects to hear a message from each of its neighbors in each round, so a deletion reduces to an erasure. 
By contrast, in the non-fully-utilized setting, insertions and deletions seem to be hard to correct. In fact, as first noted by Hoza~\cite{HozaPC15}, it seems \emph{crucial} to allow insertion and deletion errors to make the problem non-trivial. To see this, consider a party that speaks once every two rounds---on an even round to communicate the bit `0' and on odd round to denote~`1'. This communication is completely resilient to noise that only flips bits since only the timing of the transmission matters. Yet, insertions and deletions corrupt this timing-encoding and call for more sophisticated coding solutions.

\subsection{Our Contributions}
In this work 
we give an efficient interactive coding scheme with constant rate for arbitrary synchronous networks (not necessarily fully-utilized) that suffer from a certain fraction of insertion, deletion and substitution noise.
We design coding schemes for adversarial errors in the same communication model as~\cite{JKL15}; namely, over a synchronous non-fully-utilized network, where parties may or not speak in a given round, and the underlying protocol is only assumed to have a fixed speaking order. One of the most interesting aspects of our work is that our coding schemes are \emph{computationally efficient}, unlike either~\cite{JKL15} or~\cite{HS16}\footnote{These coding schemes utilized a combinatorial object known a \emph{tree code}, for which no efficient construction is known.}. 
Further, our coding scheme works on arbitrary topologies, in contrast to~\cite{JKL15}, which only applies to the star graphs. As mentioned above,~\cite{HS16} also handles arbitrary graphs, albeit in the fully utilized model.

In this paper (the first part) we assume that the parties pre-share a common random string (CRS), and assume the adversary is \emph{oblivious} (i.e. their corruptions are fixed in advance, independent of the CRS, inputs, and communication).

\begin{theorem}[Coding for oblivious noise assuming  shared randomness, informal]\label{thm:main}
Let $G=(V,E)$ be an arbitrary synchronous network with $n=|V|$ nodes and $m=|E|$ links, and assume
any two neighbours share a random string.
For any noiseless protocol $\Pi$ over~$G$ with a predetermined order of speaking, 
and for any sufficiently small constant $\eps$,
there exists an efficient coding scheme that simulates~$\Pi$
over a noisy network~$G$.  The simulated protocol is robust to adversarial insertion, deletion, and substitution noise, assuming at most  $\eps / m$-fraction of the communication is corrupted. The simulated protocol communicates $O(\CC(\Pi))$ bits, and succeeds with probability at least $1-\exp(-\CC(\Pi) /m)$, assuming the noise is oblivious.
\end{theorem}
We remark that, with $1/m$ noise rate,
the adversary can completely corrupt a single link. %
Therefore, it is natural to allow the adversary to alter at most $\eps / m$ communication, as we do above.

\gnote{2021: Added this from part 2.}
While our scheme features a constant blowup in their communication rate, the blowup in the round complexity may be up to $O(m)$; furthermore, the round complexity will \emph{not} depend on the noise, and is fixed at the beginning of the scheme. We also note that if there is no bound on the blowup in round complexity, trivial coding schemes exist, e.g., by encoding the input $x\in\{0,1\}^*$ as a single message sent at round~$x$.

In the second part of this work~\cite{GKR-2}, we show how to remove the shared randomness assumption. Further, we obtain schemes that work even when the adversarial noise is non-oblivious, albeit, 
with slightly smaller noise resilience of~$\eps/ m\log m $.

\subsection{Organization}
In the next subsections we overview the key ideas required for performing interactive coding 
in the multiparty setting with constant overhead, and discuss related work. 
Section~\ref{sec:preliminaries} fixes some notations and defines the model.
Our efficient coding scheme against an oblivious adversary assuming shared randomness
is formally described in Section~\ref{sec:protocol-obliv-crs} and analyzed in Section~\ref{sec:analysis-oblivious-crs}. Some parts of the coding scheme are based on the meeting point mechanism of~\cite{haeupler14} and are given in Appendix~\ref{app:meetingpoints} for completeness.

In the second part of this work~\cite{GKR-2},
we remove the assumption that the parties pre-share a random string, while keeping all the other properties of the coding scheme (Theorem~\ref{thm:main}). 
In addition, we eliminate the restriction to oblivious adversaries at the cost of a somewhat reduced resilience. 
Namely, building on top of the scheme presented in Section~\ref{sec:protocol-obliv-crs}, we develop in~\cite{GKR-2} an efficient coding scheme with a constant rate that is resilient against any \emph{non-oblivious} adversary that is limited to corrupting a fraction $\eps/m\log m$ of insertions and deletions.
We also develop in~\cite{GKR-2} a coding scheme that is resilient to a somewhat higher noise level of $\eps / m\log\log m$-fraction of insertion and deletion noise, while still incurring a constant blowup in the communication, by assuming again that the parties pre-share randomness.

\subsection{Coding Scheme: Key Ideas}
\label{sec:techniques}

\label{sec:high-level-motivation}
In this section we motivate the elements of our coding scheme at a high level. The basic idea towards constructing a multiparty coding scheme
is to have each pair of parties perform a two-party coding scheme~\cite{RS94,JKL15,HS16}.
However, merely correcting errors in a pairwise manner is insufficient, since if a pair of parties found an inconsistency and backtracked, this may cause new inconsistencies (in particular, between these parties and their other neighbors).
In \cite{JKL15}, this problem was solved by assuming there is one party who is connected to all parties (i.e., the star topology).
This party has a global view on the progress of the simulation
in the \emph{entire network}, and hence can function as the ``conductor of the orchestra.''

In our setting, 
no such central party exists and consequently no party knows the state of the simulation
in the entire network. Instead, each party only sees its local neighborhood, and needs to propagate its ``status'' to the entire network in a completely decentralized manner.

We mention that Hoza and Schulman~\cite{HS16} also considered an arbitrary topology, but they consider the fully-utilized model. Correcting errors efficiently in the non-fully-utilized model seems to be trickier; we elaborate on this towards the end of this section.

In order to keep our simulation \emph{efficient}, as opposed to previous works in the multiparty setting which used the (inefficient) tree-code approach, we use the rewind-if-error
approach~\cite{schulman92,BK12,KR13,haeupler14,GHKRW18} (see also~\cite{gelles17}).  Namely, each two neighboring parties send a hash of their simulated pairwise transcripts, and if the hashes do not match, then an error is detected, and the two parties initiate a ``meeting-points'' mechanism~\cite{schulman92,haeupler14}
in order to find a previous point where their transcripts agree. 

Since we want a constant-rate coding scheme with error rate up to $\Omega(1/m)$, our goal is to ensure that any corruption only causes the network to waste $O(m)$ communication. 
Indeed, if the parties need $K$ communication to correct a single error, then constant-rate coding implies a maximal noise level of~$O(1/K)$. If the noise level is higher, the communication required to correct the errors is already too high to preserve a constant-rate coding.
As a consequence of the above, in our coding schemes the parties will have to detect errors by communicating $O(1)$ bits per check, so that the consistency-checking costs $O(m)$ communication overall. In hindsight, this will restrict the parties to use hash functions with \emph{constant-sized outputs}.

\medskip
Let us now describe a preliminary attempt at a coding scheme. The coding scheme will have many \emph{iterations}, where each iteration will consist of the following. First, every pair of parties will send each other hashes of their simulated pairwise transcripts. If these match, they continue simulating with each other for a small number of rounds; else, they use the meeting points mechanism to find a common ``meeting'' point where both their partial transcripts agree and rewind their partial transcript back to that point. Until they have done so, they refuse to simulate with any other party. 
The parties continue repeating these two steps (1. Send Hashes / Do Meeting Points and 2. Simulate / Stay Silent) until they finish simulating the underlying protocol. Each step will take a fixed number of rounds, so, since we are in the synchronous model, the parties will always be able to tell which step currently is being carried out.

This naive coding scheme has a clear but important flaw. 
Once a party $u$ decides to rewind due to an error on the link $(u,v)$,
this has an effect on the simulation of~$u$ with its other neighbors.
In order to ensure correctness of the overall simulation, 
$u$ must rewind the simulation with each and every one of its neighbors
that may be affected.

It is tempting to rewind $u$'s neighborhood by having $u$ immediately truncate its simulated transcript for each adjacent link $(u,w)$, to the same place it rewound the link~$(u,v)$.
After such a truncation, a discrepancy would appear on the link~$(u,w)$ since $w$ did not change its simulated transaction, and the meeting points mechanism would resolve any such resulting discrepancy in the next iteration. 
In particular, an error on the link $(u,v)$ would cause $u$ and $v$ to perform meeting points and rewind their transcript to a consistent point. Then, in the next iteration, $u$ will be inconsistent with its other neighbors and will  initiate meeting points with all its other neighbors, resulting in them rewinding their transcripts accordingly.
In the following iteration, any neighbor $w$ of $u$ would find an inconsistency with any other party in its \emph{own} neighborhood, initiate meeting points with them, and so on until the entire network has rewound.

While the entire network eventually rewinds in this scheme, the number of iterations this takes is proportional to the diameter of the network. Recall that to ensure constant rate, we need to correct any error with at most $O(m)$ communication (with high probability), which in our case means just $O(1)$ iterations. So this naive approach could \emph{not} achieve constant rate. 
Even worse, if magically we could augment the consistency check step so that news propagates through the entire network in just a single iteration, this approach would still fail with high probability.
The consistency checking involves exchanging hashes of transcripts and verifying whether they match. \emph{Hash collisions} are a major problem here; a hash collision will cause parties $u$ and $v$ to believe that their transcripts are consistent, when really they are not. 
Since we are restricted to use constant-sized hashes to keep the rate constant, 
hash collisions occur with constant probability in each exchange. Hence, in the course of the (at least) $n-1$ hash exchanges required to inform all $n$ parties about the error, hash collisions occur with overwhelming probability. If so, with high probability many parties will not be informed of the error and will not rewind their transcript.

What saves us here is that $u$ does not need to initiate a meeting points protocol to tell $w$ to rewind---$u$ can tell $w$ to rewind \emph{directly}. The meeting points protocol is designed to find a point where the transcripts of $u$ and $w$ are consistent, when neither $u$ nor $w$ has any idea where that point may be. But in the above case, $u$ is confident that $w$ should rewind, and should simply inform~$w$ it should rewind.
To summarize the above ideas, an iteration of our coding scheme now consists of the following steps: 1. Consistency Check, 2. Simulate / Stay Silent, and 3. Rewind request (if necessary).
As before, each of these phases takes a fixed number of rounds, so the parties are always in sync.

\smallskip
Some care 
needs to be taken with the Rewind phase. Specifically, $u$ may not be able to instantly communicate to its neighbor $w$ how far it wants $w$ to rewind. 
For instance, $u$ may have rewound 
a very large number of rounds, and even telling its neighbor how many rounds to rewind would require a  large amount of communication; 
this cannot be communicated in the fixed number of rounds that the rewind phase contains. 
Instead, $u$ should simply tell $w$ to rewind a small amount at a time,
spanning the rewind process over several consecutive iterations if needed.
To summarize, the rewind phase allows the parties to gradually synchronize with each other without the use of 
unreliable hash functions, 
allowing the parties to return to simulation while circumventing the issues 
that hash collisions bring. This is one of the key conceptual observations behind our work.

\smallskip
There is one final issue that is specific to the non-fully-utilized model. 
Recall that in the non-fully-utilized model, communication may be sparse. For concreteness, consider the line network (i.e., a path graph, where party $i\in\{1,2,\ldots,n -1\}$ is connected to party~$i+1$).
Consider the example protocol, where party 1 communicates back-and-forth with party 2 for $n$ rounds while other parties remain silent, then party $2$ communicates with party $3$ for $n$ rounds while other parties remain silent, and so on, until the communication reaches parties $n-1$ and $n$, and bounces back towards party 1. The communication that occurs in any set of $n$ rounds of this underlying protocol is $n$, since only a single party is speaking at any given round. An iteration of our coding scheme would give the parties $n$ rounds to simulate, after which they perform a single consistency check.

However, when \emph{simulating} the underlying protocol in our coding scheme, the parties may not all be in agreement about which round they are simulating. For example, parties 1, 2, and 3 could think they are simulating round 1 (where parties $1$ and $2$ talk a lot), while party 4, 5, and 6 think they are simulating round $3n+1$ (where parties 4 and 5 talk a lot), and generally parties $3i+1, 3i+2,$ and $3i+3$ think they are simulating round $i \cdot n + 1$. In this case, there will be $\Omega(n^2)$ communication in just one iteration of our coding scheme (in $O(n)$ rounds). This is potentially disastrous. Not only might this lead to the communication blowup of the coding scheme being super-constant, but also the adversary would have the budget to place $\Omega(\eps n)$ additional errors, and may be able to use these errors to derail any attempts of getting the simulation back on track.

To avoid this issue altogether, we introduce a ``flag-passing'' phase in which each party informs the entire \emph{network}
whether all seems correct and the simulation should continue, or if it sees an inconsistency and the
network should idle while this is fixed. This fixes the issue above, since party 3 will realize that there is 
an inconsistency between its transcripts with parties 2 and 4, and will notify the network to avoid simulation.

Putting it all together, our resilient protocol consists of repeatedly executing the following four steps in order:
(i) consistency check,
(ii) flag passing,
(iii) simulation,
and (iv) rewind.
The coding scheme cycles through these four phases in a fixed manner, where each step takes a fixed amount of rounds to avoid ambiguity.

The formal specification of each phase, along with the formal coding scheme assuming oblivious noise and a common shared randomness is depicted as Algorithms \ref{alg:robust-protocol}--\ref{alg:flag-passing} below (see Section~\ref{sec:protocol-obliv-crs} for full details and notations). The meeting points procedure is given in Algorithm \ref{alg:meeting-points}.

\subsection{Related work}
As mentioned above, interactive coding was initiated by Schulman~\cite{schulman92,schulman96}.
Over the last several years there has been tremendous amount of work on interactive coding schemes in the two-party setting (e.g., \cite{GMS14,BR14,BE17,BKN14,GH14,GHKRW18}),  and in the multi-party setting (detailed below). We refer the reader to~\cite{gelles17} for a survey on the field (and to references therein).

In what follows we only mention the schemes that are closely related to our setting, namely, ones that are either in the multiparty setting, or ones that are in the two-party setting but are  resilient to insertions and deletions.

Coding schemes for insertions and deletions in the two-party setting were first constructed by
Braverman, Gelles, Mao, and Ostrovsky~\cite{BGMO17}.  As was noted above, in the model where in each round each party sends a single bit (which is the model used by most previous works, including~\cite{BGMO17}), insertions and deletions are only meaningful in the {\em asynchronous} model, as otherwise, such an error model is equivalent to the erasure model. Indeed,~\cite{BGMO17} considered the asynchronous model.  We note that in the asynchronous model, a single deletion can cause
a ``deadlock'', where both parties wait for the next incoming message. Therefore, Braverman et al.\@ considered a model where any deletion is followed by an insertion, thus the protocol never ``halts'' due to noise. They note that the noise may delete a certain message and then inject a spoofed ``reply'' to the original sender. In this case, one party believes that the protocol has progressed by one step, while the other party is completely oblivious to this. This type of noise was called a \emph{synchronization attack} as it brings the parties out of synch.

Braverman et al.~\cite{BGMO17} constructed a coding scheme for insertions and deletions in this model with constant communication rate, and resilience to constant fraction of noise.
Later, Sherstov and Wu~\cite{SW19} showed that a very similar scheme can actually resist an \emph{optimal} noise level.
Both these schemes are computationally inefficient.
Haeupler, Shahrasbi, and Vitercik~\cite{HSV18} constructed an efficient scheme that is resilient to (a small) constant fraction of insertions and deletions. Furthermore, they constructed a scheme where the communication rate approaches~1 as the noise level approaches~0. Efremenko, Haramaty, and Kalai~\cite{EHK20} considered the {\em synchronous} setting, where parties can send messages of arbitrary length in each round, and the adversary may insert and delete bits in the content of each message. They construct an efficient coding scheme with constant communication rate, and constant blowup in the round complexity, that is resilient to a small constant fraction of noise. All these works (\cite{BGMO17, SW19, HSV18, EHK20}) were in the two-party setting.

In the multiparty setting, Rajagopalan and Schulman~\cite{RS94} constructed a coding scheme for stochastic noise with rate $1/O(\log (d+1))$ for networks with maximal degree~$d$.
This implies a constant rate coding scheme for graphs with constant degree.
Alon et al.~\cite{ABEGH19} showed that if the topology is a clique, or a dense $d$-regular graph, then constant rate coding is also possible. Yet, Braverman, Efremenko, Gelles, and Haeupler~\cite{BEGH18} proved that a constant rate is impossible if the topology is a star. All the above works assume a synchronous fully-utilized network.
Gelles and Kalai~\cite{GK19} showed that constant rate coding schemes are impossible also on graphs with constant degree, such as a cycle, assuming a synchronous, yet not fully-utilized model.

The case of adversarial noise in the multiparty setting was first considered by Jain, Kalai, and Lewko~\cite{JKL15}, who constructed a constant-rate coding scheme over a synchronous star network that is resilient to $O(1/n)$ fraction of noise. They did not assume that the network is fully utilized, and only assumed that the underlying noiseless protocol has a fixed speaking order. 
Hoza and Schulman~\cite{HS16} considered the fully utilized model with arbitrary topology and constructed a constant rate coding scheme that is resilient to $O(1/m)$ noise. Via routing and scheduling techniques, they show how to resist a fraction of $O(1/n)$-noise, while reducing the rate to $O(n/m\log n)$.  Both these schemes use tree-codes, and therefore are computationally inefficient. Further, \cite{HS16} considers coding with different modeling assumptions such as assuming directional links or the case where the fraction of noise per link is limited.

Aggarwal, Dani, Hayes, and Saia~\cite{ADHS17} constructed an efficient synchronous coding scheme, assuming the parties use private point-to-point channels (i.e., with an oblivious adversary).
They consider a variant of the fully-utilized model, where parties are allowed to remain silent, but the receiver will hear a constant bit set by the adversary for free (or further corrupted by the adversary as any other transmission). In their model the length of the protocol is not predetermined and may vary with the noise (similar to the two-party adaptive notion of~\cite{AGS16}). Their coding scheme is resilient to an arbitrary (and \emph{a priori} unknown) amount of bit-flips (as long as the noise-pattern is predetermined and independent of parties shared randomness), and has a rate of $O(1/\log (n \CC(\Pi)))$.

Censor-Hillel, Gelles, and Haeupler~\cite{CGH19} constructed \emph{asynchronous} coding schemes, where the parties do not know the topology of the network (an assumption that is very common in the distributed computation community). Their scheme is resilient to $O(1/n)$ noise and has a rate of $O(1/n\log^2 n)$, over an arbitrary topology.

\section{Preliminaries}
\label{sec:preliminaries}

\paragraph{Notations and basic properties}
For $n\in \mathbb{N}$ we denote by $[n]$ the set $\{1,2,\dotsc,n\}$. The $\log(\cdot)$ function is taken to base~2.
For a distribution~$D$ we use $x \sim D$ to denote that $x$ is sampled according to the distribution~$D$.
For a finite set~$\Omega$ we let $\Unif_{\Omega}$ be the uniform distribution over~$\Omega$; we commonly omit~$\Omega$ and write $x\sim\Unif$ when the domain is clear from context.

\paragraph{Multiparty interactive communication model}
We assume an undirected network $G=(V,E)$ of $n=|V|$ parties, $p_1,\ldots, p_n$, and
$|E| = m$ edges, where $p_i$ is connected to $p_j$ if and only if $(p_i,p_j)\in E$.
We identify parties with nodes, and treat $p_i$ and node~$i$ as one.
For $v \in V$, let $N(v)$ denote the neighborhood of $v$ in $G$, i.e. $N(v) = \{u: (u,v) \in E \}$.
The network $G$ is assumed to be a connected simple graph (i.e., without self-loops or multi-edges).

The communication model works in synchronous rounds as follows.
At each round, any subset of parties may decide to speak. Each link is allowed to transmit at most one symbol per round in each direction from a constant-sized alphabet~$\Sigma$. We will assume throughout this paper that $\Sigma=\{0,1\}$ (both in the noiseless and noisy settings), however, our results extend to a larger alphabet as well.
At each round, a party is allowed to send multiple (possibly different) symbols over multiple links.
A \emph{transmission} over a certain link at a certain round is the event of a party sending a message on this link at that round (if both parties send messages these are two separate transmissions).

We emphasize that, contrary to most previous work, our communication model is not fully-utilized and does not demand all parties to speak at each round on every communication channel connected to them; in fact we don't demand a certain party to speak at all at any given round.

\paragraph{Multiparty protocol}
Each party is given an input $x_i$, and its desire is to output $f_i(x_1,\ldots,x_n)$ for some predefined function~$f_i$ at the end of the process.
An interactive  \emph{protocol} $\Pi$ dictates to each party what is the next symbol to send over which channel (if any), as a function of the party's input, the round number, and all the communication that the party has observed so far. After a fixed and predetermined number of rounds, the protocol terminates and each party outputs a value as a function of its input and observed transcript.
The \emph{length of the protocol}, also called its \emph{round complexity} $\RC(\Pi)$
is the maximal number of rounds $\Pi$ takes to complete on any possible input. 
The communication complexity of the protocol (in bits), denoted by $\CC(\Pi)$, is the total number of transmissions in the protocol times $\log |\Sigma|$. Since we assume $\Sigma=\{0,1\}$, the communication complexity equals the number of transmissions.

\paragraph{Noise model}
We concern ourselves with simulating a multiparty interactive protocol~$\Pi$ among $n$ parties over a \emph{noisy} network. 
To do so, we run a different protocol, $\widetilde \Pi$, whose task is to compute the transcript of~$\Pi$ in the noisy network.

A single communication event over a noisy channel with alphabet $\Sigma$ is defined using a symbol from $\Sigma \cup \{ * \} $ that is sent on the channel, where $*$ is a special symbol that means ``no message''.

In this part of our work, we concern ourselves on a restricted type of noise, namely, \emph{oblivious noise} or \emph{oblivious adversary}.\footnote{In contrast, part II of our work~\cite{GKR-2} deals with  general, non-oblivious noise.}
An oblivious adversary is an adversary that pre-determines its noise attack, independently of the inputs and randomness of the parties. Namely, recall that the round complexity of our protocols is fixed. An oblivious adversary is one that pre-determines, at the onset of the specific instance of the protocol, which noise it will put on each link and each round. %

We consider two types of 
oblivious adversaries: \emph{fixing} adversaries and \emph{additive} adversaries. 
Assuming a binary alphabet, an oblivious fixing adversary fixes a noise pattern 
$e={\{0,1,*,\perp\}}^{2|E| \cdot \RC(\widetilde\Pi)}$.
In particular, the entry $e_{i,(u,v)} \in \{0, 1, *, \perp\}$ determines whether noise occurs on the link $u\to v$ in round $i$
and, if so, the new message on the link $u\to v$ in round $i$. 
We interpret $\perp$ as leaving the communication
as-is, $*$ as fixing no message, and 0 and 1 as fixing 0 and 1 on the link, respectively.
Note that $\RC(\widetilde\Pi)$ might depend on $n,m$ and $\Pi$ but it is assumed to be independent of the specific inputs and the specific noise, i.e., any execution takes the same amount of rounds.
The number of errors is the number of non-$\perp$ entries of $e$.\footnote{Note, this might be somewhat unintuitive since the adversary is charged an error even if its interference did not eventually change the output symbol of the channel.}

We also consider additive adversaries~\cite{BBT60} (see, e.g., \cite{CDFPW08,GS10,GIPST14} for applications). 
An oblivious additive adversary fixes a noise pattern $e=\{0,1,2\}^{2|E| \cdot \RC(\widetilde\Pi)}$ 
that defines the noise per each link in each round of the protocol. 
In particular, the entry $e_{i,(u,v)}\in \{0,1,2\}$ determines the noise added to the link $u\to v$ in round $i$. 
Assuming $u$ transmits to~$v$ in iteration~$i$ the message  $t\in\{0,1,2\}$ (where $2$ denotes the case of no message, i.e., $*$), then $v$ receives the transmission~$t+e_{i,(u,v)}\mod 3$.
The number of errors is the number of non-zero entries in~$e$.

While our results in this work (i.e., the first part) 
are stated for fixing and additive adversaries, 
they in fact apply to any oblivious adversary with a sufficiently limited number of errors.
That is, our results apply to \emph{any} (oblivious) adversary that commits to error functions $\ch_{i,(u,v)}: \Sigma \cup \{ * \} \to \Sigma \cup \{ * \}$ for each link $(u,v)\in E$ and each round~$i\in \RC(\widetilde\Pi)$, such that the adversary is charged an error whenever the function is not the identity.\footnote{We thank the anonymous referee for suggesting this general description.} 

In the presence of noise, an instance (i.e., a specific run) of some protocol~$\widetilde\Pi$ might behave differently than an instance when the channels commit no errors. 
While the round complexity of all our schemes are fixed as mentioned above,
the communication in a noisy instance of a protocol may vary
as a function of the adversary's noise: the noise can make parties think it is their turn to talk when it is not, and vice versa. Hence, in the noisy setting we define $\CC(\widetilde\Pi)$ to be the communication complexity of a \emph{given instance} of the protocol~$\widetilde\Pi$, and this quantity is determined by the inputs, the randomness, and the specific noise pattern observed in this instance. The communication complexity of a noisy instance instance of~$\pi$ is defined to be the number of bits that are \emph{sent} by parties in~$\pi$ (namely, the communication is counted even if the message is deleted, and not counted for an inserted message).

Moreover, we emphasize that, opposed to the fully utilized model where the number of rounds fixes the communication complexity, in our model these two are only related by the trivial bound $\CC(\widetilde\Pi) \le 2|E| \log |\Sigma|\cdot \RC(\widetilde\Pi)$.
Furthermore, for our protocols this bound can be very loose.

The \emph{fraction of noise} observed in a given instance,
is the fraction of corrupt transmissions out of all the transmissions in that instance.
For the binary case the noise fraction can be written as,
\[
\mu = \frac{\#\text{errors}}{\CC(\widetilde\Pi)}.
\]
This is also known as \emph{relative} noise. Perhaps unintuitively, there are trivial examples where the relative noise can be greater than 1; for example, a completely silent protocol in which the adversary inserts a message. However, we will only concern ourselves with relative noise much less than 1 in this paper. See also~\cite{AGS16} for a similar notion of relative noise in the two-party setting.

\begin{remark}\label{rem:additive}
We prove our results for both oblivious fixing and additive adversaries. As far as oblivious adversaries go, 
the fixing adversary seems more natural to us, and so for the purposes of this paper (Part~I) in isolation, we find the result 
most interesting with the fixing adversary. However, our end goal is to get a result for \emph{non-oblivious} adversaries, 
which we do in Part~II~\cite{GKR-2}. To make our proof in~\cite{GKR-2} simpler, we need to prove resilience against the oblivious 
\emph{additive} adversary in this paper.

We emphasize that 
for both adversaries we have described, the number of errors they are charged is fixed
the moment they fix their error pattern. Specifically, the number of errors the adversary is \emph{deterministic}
once the adversary chooses a pattern---it does not depend on any randomness the parties use.
\end{remark}

\paragraph{Coding scheme---a noise-resilient protocol} 

A coding scheme is a scheme that converts any protocol~$\Pi$ into a noise-resilient protocol~$\widetilde\Pi$ that simulates~$\Pi$ correctly over the noisy network with high probability on any given input.
We say that a protocol $\widetilde \Pi$ simulates~$\Pi$ correctly on a given input if each party can obtain its output corresponding to~$\Pi$ from the transcript it sees when executing~$\widetilde \Pi$.
Informally, the protocol $\widetilde \Pi$ is said to be resilient to $\mu$-fraction of noise (with probability~$p$),
if it simulates $\Pi$ correctly (with probability at least~$p$)
when executed over a noisy network with adversarial noise that commits up to $\mu \CC(\widetilde \Pi)$ errors.

To be more precise, recall that in our model the number of transmissions depends on the adversarial errors and the randomness that the parties use. So, $\CC(\widetilde \Pi)$ is not known a priori to the adversary or to the parties. Hence, we somewhat relax this condition; we say that $\widetilde \Pi$ is resilient to $\mu$ fraction of noise if, for \emph{any} oblivious adversary and any inputs for the parties, the probability of the event that the adversary commits at most $\mu \CC(\widetilde \Pi)$ errors and $\Pi$ is not simulated correctly, is at most~$1-p$.

We will assume that the noiseless protocol to be simulated, $\Pi$, has the property that the speaking order is independent of the inputs that the parties receive. Note that this assumption restricts generality.
However, 
we emphasize that this requirement does not apply on the coding scheme that simulates~$\Pi$ over the noisy network.\footnote{This decision resembles the model of~\cite{JKL15}.}
We will usually use the notation $|\Pi|$ to denote the length of the (noiseless) protocol \emph{in chunks} rather than in rounds; see Section~\ref{sec:chunks} for details on partitioning protocols into chunks.

\paragraph{Hash functions}

We use an inner-product based hash function. The hash function is seeded with a random string $s$ such that each bit of the output is an inner product between the input $x$ and a certain part of $s$ (using independent parts of $s$ for different output bits). Formally,

\begin{definition}[Inner Product Hash Function]
\label{def:ip-hash}
The inner product hash function 
$h: \{0,1\}^* \times \{0,1\}^* \to \{0,1\}^\tau$ 
is defined for any input $x$ of length $|x|=L$ and seed~$s$ of length $|s|=\tau L$ (for $\tau \in \mathbb{N}$), as the concatenation
of $\tau$ inner products between~$x$ and disjoint parts of~$s$, namely,
\[
 h(x,s) = \langle x, s[1,L] \rangle \circ \cdots \circ  \langle x, s[(\tau-1)L+1,\tau L] \rangle.
\]
We use the shorthand $h_s(x) \eqdef h(x,s)$.
\end{definition}
We sometimes abuse notation and hash a ternary-string $x$ (or a string over a larger alphabet). In this case, assume we first convert $x$ into a binary string in the natural manner (each symbol separately, using $\lceil \log_2 3\rceil =2$ bits) and then hash the binary string. The seed length should increase appropriately (by at most a constant).

The following is a trivial property of this hash function, stating that, given a uniform seed, the output is also uniformly distributed.
\begin{lemma}\label{lem:unifHash}
For any $\tau, L \in \mathbb{N}$, $x \in \{0,1\}^{L} \setminus \{0\}$, and any $r\in\{0,1\}^\tau$
\[
\Pr_{s\sim \Unif} [ h_s(x) = r] = 2^{-\tau}.
\]
where $\Unif$ is the uniform distribution over $\{0,1\}^{\tau \cdot L}$.
\end{lemma}
It is easy to see Lemma~\ref{lem:unifHash} also implies that the collision probability of the inner product hash function with output length $\tau$ is exactly $2^{-\tau}$, since given two strings $x$ and $y$ (of the same length)
such that $x \neq y$, the Lemma implies that the probability that $\Pr_{s\sim \Unif} [ h_s(x-y) = \vec{0}] = 2^{-\tau}$.

\section{Coding scheme for oblivious adversarial channels}
\label{sec:protocol-obliv-crs}
\subsection{Overview}
\label{sec:protocol-overview}
The high-level description of the simulation is as follows.
The basic mechanism is the rewind-if-error approach from previous works \cite{schulman92,BK12,haeupler14} (see also~\cite{gelles17}). 
In particular, the parties execute the noiseless protocol~$\Pi$ for some rounds and then exchange some information to verify if there were any errors. If everything seems consistent, the simulation proceeds to the next part; otherwise, the parties rewind to a previous (hopefully consistent) point in $\Pi$ and proceed from there.

Note that since multiple parties are involved, it may be that some parties believe the simulation so far is correct while others believe it is not.
Yet, even if one party notices an inconsistency, the entire network may need to rewind. Hence, we need a mechanism that allows propagating the local view of each party to the entire network.

Our simulation algorithm consists of repeatedly executing the following four phases:
(i) consistency check,
(ii) flag passing,
(iii) simulation,
and (iv) rewind.
The simulation protocol cycles through these four phases in a fixed manner, and each such cycle is referred to as an  {\em iteration}.  Each phase consists of a fixed number of rounds (independent of the parties' inputs and the content of the messages exchanged).  Therefore, there is never an ambiguity as to which phase (and which iteration) is being executed. 
We next describe each phase (not in the order they are performed in the protocol).

\begin{itemize}
\item [(i)] \textbf{Simulation:}
In this phase the parties simulate \emph{a single chunk} of the protocol $\Pi$. Specifically,
we split $\Pi$ into chunks---consecutive sets of rounds---where at each chunk $5\K$ bits are being communicated, for some $\K \geq m$ that is fixed throughout the simulation 
and such that $\K$ is divisible by $m$.  Jumping ahead, we note that $\K$ is set to be $m$ in the protocol we construct in this paper, which only considers to oblivious adversaries, and is set to $m\log m$ in our protocol in~\cite{GKR-2} which considers arbitrary (non-oblivious) adversaries.
Note that since the speaking order in~$\Pi$ is fixed and predetermined, the partition into chunks is independent of the inputs and can be done in advance. We assume without loss of generality that each party speaks at least once in each chunk (this is without loss of generality since one can preprocess $\Pi$ to achieve this property while increasing the communication complexity by only a constant factor).

In this phase, the parties ``execute'' the next chunk of $\Pi$, sending and receiving messages as dictated by the protocol~$\Pi$.

This phase \emph{always} takes $5\K$ rounds, which is the maximal number of rounds required to simulate $5\K$ transmissions of~$\Pi$. It may be that the simulation of a specific chunk takes fewer rounds; in this case, the phase still takes $5\K$ rounds where all the parties remain silent after the chunk's simulation has completed until $5\K$ rounds have passed.

We note that some parties may be aware that the simulation so far contains errors that were not corrected yet (jumping ahead, this information can be obtained via local consistency checks that failed or via the global flag-passing phase, described below).
When we reach the simulation phase, these parties will send a dummy message~$\bot$ to their neighbors and remain silent for $5\K$ rounds until the simulation phase completes.

\item [(ii)] \textbf{Consistency check:}
The main purpose of this phase is to check whether each two \emph{neighboring parties} $u,v\in V$
have consistent transcripts and can continue to simulate, or whether instead they need to correct prior errors.
This phase is based on the \emph{meeting points} mechanism~\cite{schulman92,haeupler14},
which allows the parties to efficiently find the highest chunk number up to which they both agree.

Roughly speaking, every time the parties enter this phase, they exchange a hash of their current transcripts with each other. If the hashes agree, the parties believe that everything is consistent and effectively continue with simulating $\Pi$.
If the hashes do not agree, the parties try to figure out the longest point in their transcript where they do agree.
To this end, they send hashes of prefixes of their transcript until the hashes agree. In our setting, each time the parties
enter the ``consistency check'' phase they perform a \emph{single} iteration of the meeting-points mechanism~\cite{haeupler14}, which consists of sending two hash values. If the hashes mismatch, they will send the next two hash values (of some prefixes of the transcript, as instructed by the meeting-points mechanism) \emph{next time they enter the consistency check phase}.\footnote{In addition to exchanging two hash values corresponding to prefixes of the transcript, the parties also exchange a hash indicating how long they have been running the meeting-points mechanism; see Section~\ref{sec:meeting-points} for a full description.}

Note that the above is performed between
each pair of adjacent parties, in parallel over the entire network.

\item [(iii)] \textbf{Flag passing:}
In the flag passing phase, the parties attempt to synchronize whether or not they should continue the simulation of $\Pi$ in the next simulation phase. As mentioned, it may be that some parties believe that the simulation so far is flawless while others may notice that there are some inconsistencies. In this phase the information about (2-party) inconsistencies is propagated to all the parties.

Roughly, if any party believes it shouldn't continue with the simulation, it notifies all its neighbors, which propagate the message to the rest of the network, and no party will simulate in the upcoming simulation phase.
However, if all parties believe everything is consistent then no such message will be sent, and all the parties will continue simulating the next chunk of~$\Pi$.

Technically speaking, the parties accomplish this synchronization step by passing a ``flag'' (i.e., a stop/continue bit) along a spanning tree $\mathcal{T}$ of~$G$. Namely, each party receives flags from each of its children in~$\mathcal{T}$. If one of the flags is stop, or if the party sees inconsistency with one of its neighbors, it sends a stop flag to its parent in the tree. Otherwise, it sends its parent the continue flag.
After the root of $\mathcal{T}$ receives all the flags, the root propagates the computed flag in the opposite direction back to the leafs. If there is no channel noise in this phase, it is clear that all parties are synchronized regarding whether the simulation should continue or not (recall that a party sends a dummy message during the simulation phase if its flag is set to stop).

\item [(iv)] \textbf{Rewind:}
In the rewind phase, each party tries to correct any obvious (i.e., length-wise) inconsistencies with their neighbors.
Recall that the meeting-points mechanism allows two neighboring parties to truncate their mis-matching transcripts to a prefix  on which both parties agree. However, this may cause inconsistencies with all their other neighbors.
Indeed, if $u$ and $v$ rewind several chunks off their transcript with each-other, then $u$ must inform any other party $z\in N(u)$ to rewind the same amount of chunks. This rewinding happens \emph{even if the transcripts on the link $(u,z)$ are consistent} at both ends.

Technically, if  the transcript of $u$ and $v$ consists of $k$ chunks, then
$u$ will send a ``rewind'' message to any neighbor $z$ for which the transcript of $u$ and $z$ contains more than $k$ chunks.
However, there are a few caveats. First, any party $z$ that is currently trying to find agreement with $u$ via the meeting-points mechanism should not rewind the transcript with~$u$. Intuitively, we can see that any such rewind seems unnecessary, since $z$ is already going to truncate its transcript when it eventually finds agreement with $u$ in the meeting-points subroutine.
Furthermore, an underlying assumption of the meeting-points protocol is that, until the parties decide to truncate their transcripts in the protocol, their transcripts do not change.
This is crucial since the MP mechanism may span over many rounds, and any abort (caused by a rewind) may make all this progress void.

Additionally, we restrict each party to rewinding at most one chunk in each of its pairwise transcripts. This is primarily for ease of analysis: it means that no matter what kind of errors the adversary induces, there is only so much harm that can be done during the rewind phase. The upshot of this is that it is not necessarily true that after the rewind phase, $u$ sees exactly the same amount of simulated chunks with all $z\in N(u)$.

Once a party $u$ sends a rewind message to a neighbor~$z$, party~$z$ will truncate one chunk of the transcript that corresponds to the link~$(u,z)$, and might then want to send rewind messages to its own neighbors.
These rewinds could trigger more rewinds, leading to a wave of rewinds going through the network.
By providing $n$ rounds in the rewind phase, we make sure that this wave has enough time to go through the entire network.\footnote{Alternatively, we could have fixed the rewind phase to consist of $D$ rounds (rather than~$n$ rounds), where $D$ is the the diameter of the graph~$G$.} This is critical to guaranteeing that we fix errors quickly enough to simulate $\Pi$ with constant overhead.
\end{itemize}

\begin{algorithm}[htp]
\caption{A noise-resilient simulation of $\Pi$ (for party $u$)}
\label{alg:robust-protocol}
\begin{algorithmic}[1]
\small
\Statex 
Let $T_{u,v}$ denote the partial, pairwise transcript between $u$ and $v$ \emph{according to $u$}, and let $|T_{u,v}|$ denote the number of \emph{chunks} simulated so far in $T_{u,v}$.
\Statex
\State \Call{InitializeState}{\,}
\Statex
\For{$i=1$ to $100|\Pi|$}
  \ForAll{$v \in N(u)$ in parallel}	  \Comment{\textbf{meeting points}}
    \State $status_{u,v} \gets \Call{MeetingPoints}{u,v, S_{i,u,v}}$ \label{step:mp-called}
  \EndFor
  \State $minChunk \gets \min_{v \in N(u)} |T_{u,v}|$
  \If{exists $v$ such that $status_{u,v} = \text{``meeting points''}$}
    \State{$status_u \gets 0$}
  \ElsIf{exists $v$ such that $|T_{u,v}| > minChunk$}\label{step:minChunkCheck}
    \State{$status_u \gets 0$}
  \Else
    \State{$status_u \gets 1$}
  \EndIf
  \Statex \Comment{\textbf{flag passing}}
  \State $netCorrect_u \gets \Call{FlagPassing}{u, status_u}$ 
  \Statex
  \If{$netCorrect_u = 1$} \Comment{\textbf{simulation}}
    \State Listen for one round. \label{step:listen-bot}
    \State \parbox[t]{0.8\columnwidth}{Simulate chunk $|T_{u,v}|+1$ with each party $v \in N(u)$ from whom we have not received $\perp$ in line~\ref{step:listen-bot}. The simulation is based on the partial transcript $T_{u,w}$ for each $w \in N(u)$, as well as the input to $u$. \strut \rule[-.4\baselineskip]{0pt}{\baselineskip}}
     \State \parbox[t]{0.8\columnwidth}{If the above step took less than $5\K$ rounds, wait until $5\K$ rounds have passed.\strut}
    \If{received no $\perp$'s in Line~\ref{step:listen-bot} in this iteration}
      \State $minChunk \gets minChunk + 1$
    \EndIf
  \Else
    \State Send a single $\perp$ to each neighbor, and wait $5\K$ rounds.
  \EndIf
  \Statex
  \For{round $r=1$ to $n$} \Comment{\textbf{rewind}}
    \ForAll{$v \in N(u)$ in parallel}
      \If{{$status_{u,v} \neq \text{``meeting points''}$ \Statex\hspace{\dimexpr \algorithmicindent * 4} and $alreadyRewound_{u,v}=0$}}
        \If{$|T_{u,v}| > minChunk$} \label{step:truncate}
          \State Send a rewind message to $v$.
          \State truncate $T_{u,v}$ by one chunk.
          \State  $minChunk \gets \min_{v \in N(u)} |T_{u,v}|$
          \State{$alreadyRewound_{u,v} \gets 1$}
        \EndIf
      \EndIf
      \If{a rewind message is received from $v$}
        \If{$status_{u,v} \neq$ ``meeting points'' \Statex\hspace{\dimexpr \algorithmicindent * 5} and $alreadyRewound_{u,v}=0$}
          \State Truncate $T_{u,v}$ by one chunk.
          \State $minChunk \gets \min_{v \in N(u)} |T_{u,v}|$
          \State{$alreadyRewound_{u,v} \gets 1$}
        \EndIf
      \EndIf
    \EndFor
  \EndFor
\EndFor
\end{algorithmic}
\end{algorithm}

\begin{algorithm}[htp]
\caption{InitializeState()}
\label{alg:initialize-state}
\begin{algorithmic}[1]
\small
\State $\K \gets m$
\ForAll{neighbors $v \in N(u)$}
  \State Initialize $T_{u,v} = \emptyset$
  \State $status_{u,v} \gets$ ``simulate''
  \State $alreadyRewound_{u,v} \gets 0$
  \State $S_{u,v} := (S_{i,u,v})_{i \in [100|\Pi|]} \stackrel{\text{unif}}{\gets}  \left(\{ 0,1\}^{\Theta(|\Pi|\K)}\right)^{100|\Pi|}$ \Statex \hspace{\algorithmicindent}uniform bits of randomness.
\EndFor
\State $status_u \gets 1$.
\State $netCorrect_u \gets 1$
\end{algorithmic}
\end{algorithm}

\begin{algorithm}[htp]
\caption{FlagPassing($u$, $status$)}
\label{alg:flag-passing}
\begin{algorithmic}[1]
\small
\Statex
Let $\mathcal{T}$ be a spanning tree of~$G=(V,E)$ with a root $\rho\in V$, known to all parties. 
Denote by $d(v)$ the distance (in edges) of $v\in V$ from~$\rho$.  
Let $\depth(\mathcal{T)}=\max_v d(v)$ be the depth of $\mathcal{T}$.

 \Statex
\Statex  // convergecasting the status to the root~$\rho$.
\State Wait until round $r= \depth(\mathcal{T})-d(u)+1$. 
\State $netCorrect \gets status \wedge \left( \bigwedge_i s_i \right)$, where $\{s_i\}$ are messages received in round $r-1$ from $u$'s children (if any).
\State Send $netCorrect$ to parent.

 \Statex
\Statex  // downcasting the result to the entire~$\mathcal{T} $
\State Wait until round $r'=r+2d(u)$.
\State $netCorrect \gets status \wedge s_p$, where $s_p $ is the message received in round $r'-1$ from parent (if any).
\State Send $netCorrect$ to children.
\Statex
\State \Return $netCorrect$
\end{algorithmic}
\end{algorithm}

\subsection{The coding scheme}
We now formally describe the coding scheme assuming a common random string (CRS) and oblivious noise (Algorithm~\ref{alg:robust-protocol}).
Let $\Pi$ be a noiseless protocol over $G=(V,E)$, with $R$ rounds and $C$ transmissions  throughout.
Assume that the communication pattern and amount is predetermined, and independent of the parties' inputs and the transcript.
Namely, let $T_\Pi = (m_1 m_2 \cdots m_{C})$ be the noiseless transcript of $\Pi$; the content of the messages $m_i$ depend on the specific inputs, however their order, timing, source and destination are fixed for~$\Pi$.

We partition $T_\Pi$ into rounds according to $\Pi$, and group the rounds into \emph{chunks}, where each chunk is a 
set of contiguous rounds with total communication complexity exactly $5\K$.
Specifically, we keep adding rounds to a chunk until adding a round would cause the communication to exceed~$5\K$. Note that without the last round, the communication in the chunk is at least $5\K - 2m + 1$ bits.
We can then add a virtual round that makes the communication in the chunk be exactly~$5\K$ bits.
This addition affects the communication complexity by a constant factor. From this point on, we assume that~$\Pi$ adheres to our required structure.

We number the chunks in order, starting from 1.
For any (possibly partial) transcript $T$, we let $|T|$ denote the \emph{number of chunks} contained in the transcript~$T$. In particular, $|\Pi|$ is the number of chunks in~$\Pi$.
\label{sec:chunks}  %
We assume without loss of generality that in each chunk, each party sends at least one bit to each of its neighbors (again, this can easily be achieved by pre-processing~$\Pi$ while increasing its communication by a constant factor).  In addition, we assume that the protocol $\Pi$ is padded with enough dummy chunks where parties simply send zeros. Concretely, we assume the parties simulate a protocol~$\hat{\Pi}$, where $\hat{\Pi}$ is formed by just taking $\Pi$ and adding $99|\Pi|$ extra chunks of padding. This padding is standard in the literature on interactive coding, and is added to deal with the case that the adversary behaves honestly in all the rounds until the last few rounds, and fully corrupts the last few rounds. For the remainder of the paper, we omit explicit mention of $\hat{\Pi}$, with the understanding that the parties are running the simulation on this padded version of $\Pi$.

The parties simulate $\Pi$ one chunk at a time, by cycling through the following phases in the following order:
consistency check, flag passing, simulation, and rewind.
Each phase takes a number of rounds that is a priori fixed, and since our model is synchronous, the parties are always in agreement regarding which phase is being executed.

Let $T_{u,v}$ denote the pairwise transcript of the link $(u,v)$ as seen by~$u$, where $v \in N(u)$; similarly,
$T_{v,u}$ is the transcript of the same link as seen by~$v$ (which may differ from $T_{u,v}$ due to channel noise).
In more detail, $T_{u,v}$ is the concatenation of the transcripts generated at each chunk, where the transcript of chunk~$i$
consists of two parts: 
(1) the simulated communication of chunk~$i$, and (2) the chunk number~$i$.\footnote{It is important to add the chunk number since the inner-product hash function we use (Lemma~\ref{def:ip-hash}) has the property that for any string $x$, $h(x)=h(x\circ 0)$.}
The structure of part (1) is as follows.
Assume that in the $i$-th chunk in $\Pi$,  $j$ bits are exchanged over the $(u,v)$ link in rounds $t_1,\ldots, t_j$.
Then $T_{u,v}$ holds a string of length $j$ over $\{0,1,*\}$ describing the communication
at times $t_1,\ldots, t_j$, \emph{as observed by $u$}. 
The symbol $*$ denotes the event of not receiving a bit at the specific round (i.e., due to a deletion). 
The transcript $T_{v,u}$ is defined analogously from $v$'s point of view. 
Note that restricted to the substrings that belongs to chunk~$i$, $T_{u,v}=T_{v,u}$ if and only if there where no errors
at rounds $t_1,\ldots,t_j$ in the simulation phase; insertions and deletions at other rounds are ignored.
We abuse notation and define $|T_{u,v}|$ to be the number of \emph{chunks} that appear in~$T_{u,v}$.

\medskip

In Algorithm~\ref{alg:robust-protocol} we describe the noise-resilient protocol for a fixed party $u$.
The parties start by initializing their state with a call to InitializeState()  (Algorithm~\ref{alg:initialize-state}).

Next, the parties perform a single iteration of the meeting-points mechanism  (Algorithm~\ref{alg:meeting-points} in Appendix~\ref{sec:meeting-points}).
Given a pair of adjacent parties $u$ and $v$, the meeting-points mechanism outputs a variable $status_{u,v}$, which indicates whether the parties want to simulate (in which case $status_{u,v} = \text{``simulate''}$) or continue with the meeting-points mechanism (in which case $status_{u,v} = \text{``meeting points''}$).

Then, according to the output of the meeting-point mechanism and according to any apparent inconsistencies in the simulated transcripts with its neighbors,
each party sets its ``flag'' $status_u$ to denote whether it should continue with the simulation or not.
This status is used as an input to the flag-passing phase, described in Algorithm~\ref{alg:flag-passing}, where $status_u = 0$ denotes that a ``stop'' flag should be propagated.

Each party ends the flag-passing phase with a flag denoted $netCorrect_u$ that is set to~$1$ if the network as a whole seems to be correct.
Then, the parties perform a simulation phase. If the $netCorrect$ flag is set to~$1$, they execute $\Pi$ for one additional chunk, according to the place they believe they are at. Otherwise, they send a special symbol~$\bot$ to indicate they are not participating in the current  simulation  phase.

Finally, the rewind phase begins, where any party $u$ that sees an obvious discrepancy in the lengths of the transcripts in its neighborhood, sends a single rewind request to any neighbor $v$ which is ahead of the rest, conditioned that $u$ and $v$ are not currently in the middle of a meeting-points process.

\section{Coding scheme for oblivious channels: Analysis}
\label{sec:analysis-oblivious-crs}

In this section we analyze the coding scheme presented in Section~\ref{sec:protocol-obliv-crs} and prove the following theorem.

\gnote{Read here 2021}\gnote{The difference between what this theorem says (assuming an adversary of rate ...) versus what we actually prove (that with high probability, either the parties simulate correctly OR the adversary commits a large fraction of errors) has been quite jarring to our reviewer. Can we please change this?}

\begin{theorem}\label{thm:oblivious-adv-crs}
Let $G=(V,E)$ be a network with $n=|V|$ parties and $m=|E|$ links, 
and let $\Pi$ be a multiparty protocol on the network~$G$ with communication complexity $\CC(\Pi)$, binary alphabet and fixed order of speaking.
Let $|\Pi| := \frac{\CC(\Pi)}{5m}$ and let $\eps > 0$ be a sufficiently small constant.
Assume an instance of  Algorithm~\ref{alg:robust-protocol} in which an oblivious adversary (fixing or additive) makes $\Err$ errors, and let $\CC$ be the communication complexity of this instance.

Then, with probability at least $1 - \exp(-\Omega(|\Pi|))$,
either Algorithm~\ref{alg:robust-protocol} simulates $\Pi$ correctly with $\CC=\Theta(\CC(\Pi))$, or $\Err > (\eps / m)\CC(\Pi)$.
\end{theorem}
\ran{Changed Theorem's statement. Please verify}
We emphasize that Theorem~\ref{thm:oblivious-adv-crs} applies to both fixing and additive oblivious adversaries (see Section~\ref{sec:preliminaries} for definitions). In fact, all 
we need here is for the transcripts to be independent of the seeds used to hash them; see Section~\ref{sec:bound-hash-overview} for 
an expanded discussion.

In order to prove the above theorem we define a \emph{potential function} that measures the progress of the simulation at every iteration. In Section~\ref{sec:potential} we define the potential function and  intuitively explain  most of its terms.
In Section~\ref{sec:potential-increase} we prove that in every iteration\footnote{Recall that a single iteration of \Alg consists of a consistency phase, flag-passing phase, simulation phase and a rewind phase.}  the potential increases by at least~$\K$, while the communication increases by \emph{at most} $\K \times \ell$, where $\ell$ measures the number of channel errors and hash collisions that occurred in that specific iteration.

We split the analysis of the potential into two parts: the meeting points mechanism and the rest of the coding scheme. The first part, in Section~\ref{subsec:MP}, re-iterates the analysis of~\cite{haeupler14} with minor adaptations. We defer the full proofs to Appendix~\ref{sec:meeting-points}. The rest of the potential analysis is novel and performed in  Sections~\ref{sec:potential-rises-no-err} and~\ref{sec:potential-rises-err-mp}.  Specifically, in  Section~\ref{sec:potential-rises-no-err} we focus on the iterations with no errors/hash-collisions, and in Section~\ref{sec:potential-rises-err-mp} we focus on iterations that suffer from errors/hash-collisions.
Then, in Section~\ref{sec:bounding-collisions} we bound (with high probability) the number of hash-collisions that may happen throughout the entire execution of the coding scheme.
Finally, in Section~\ref{sec:proof-main-thm} we complete the proof of Theorem~\ref{thm:oblivious-adv-crs}, by showing that
the potential at the end of the coding scheme must be high enough to imply a correct simulation of~$\Pi$, given the bounded amount of errors and hash-collisions.

In the following, all our quantities measure progress in chunks, where each chunk contains exactly $5\K=5m$ bits.
Recall that we denote by~$|\Pi|$ the number of chunks in the noiseless protocol~$\Pi$, and we denote by $|T_{u,v}|$ the number of \emph{chunks} in the simulated (partial) transcript~$T_{u,v}$.

\subsection{The potential function}\label{sec:potential}
Our potential function~$\phi$ will measure the progress of the network towards simulating the underlying interactive protocol $\Pi$ correctly. Naturally, $\phi$ changes as 
the simulation of \Alg 
progresses, and so depends on the round number. In what follows, for ease of notation, we omit the current round number in all the terms used to define~$\phi$.

For each adjacent pair of parties $u$ and $v$, define
\begin{equation}\label{eq:def:G}
G_{u, v}
\end{equation}
to be the size (in chunks)
of the longest common prefix of $T_{u,v}$ and $T_{v,u}$.
Namely, $G_{u, v}$ is the length of the largest prefix of communication between parties $u$ and $v$ in~$\Pi$, that these parties agree on.
Define $B_{u, v}$ to be
\begin{equation}\label{eq:def:B}
B_{u, v} \eqdef \max(|T_{u,v}|, |T_{v,u}|) - G_{u, v}.
\end{equation}
Namely, $B_{u,v}$ is the gap between how far one of the parties \emph{thinks} they have simulated and how far they have simulated correctly.\footnote{Note that we can have $B_{u,v} = 0$ even when there have been errors in the network, as long as those errors were corrected.}
Note that $B_{u,v}$ is always nonnegative by design. Furthermore, $B_{u,v}=0$ if and only if the parties have no differences in their pairwise transcripts with each other.

Define 
\begin{equation}\label{eq:def:G*}
G^* \eqdef \min_{(u,v) \in E}G_{u, v}
\end{equation}
to be the largest chunk number through which the network \emph{as a whole} has correctly simulated.
Let 
\begin{equation}\label{eq:def:H*}
H^* \eqdef \max_{u}\max_{v \in N(u)} |T_{u,v}|
\end{equation}
denote the largest chunk number which any party in the network \emph{thinks} it has simulated; note that, by definition, $H^* \geq G^*$.
Finally, we define 
\begin{equation}\label{eq:def:B*}
B^* \eqdef H^* - G^*.
\end{equation}

In addition, our potential function also quantifies the progress of the meeting-points mechanism between any two adjacent parties in the network (which we elaborate on in Section~\ref{subsec:MP} below, and in Appendix~\ref{sec:meeting-points}).
This is done via the term $\varphi_{u,v}$ defined in Eq.~\eqref{eqn:Hpot} in Section~\ref{sec:app-mp-potential},
which is closely inspired by the potential function stated in~\cite{haeupler14}.
Intuitively, $\varphi_{u,v}$ is the number of iterations of the meeting-points mechanism that parties $u$ and $v$ need to do to make $B_{u,v} = 0$; indeed, for all pairs $(u,v) \in E$ it holds that  (Proposition~\ref{prop:phi-nonnegative})
\[ 0 \leq B_{u,v} \leq \varphi_{u,v},\]
and in particular, $\varphi_{u,v} = 0$ implies that $B_{u,v} = 0$ .

Finally, let $\EHC$
denote the number of errors and hash collisions that have occurred in the protocol until the current round of \Alg (hash collisions may occur in the meeting-points mechanism in Appendix~\ref{sec:meeting-points}). Similarly to all the other terms in the potential, we drop the dependence on the round~$r$.

Our potential function is defined to be:
\begin{equation}\label{eqn:potential}
\phi \eqdef \sum_{(u,v) \in E} \left( \frac{\K}{m}G_{u,v} - \K \cdot \varphi_{u,v}\right) - C_1 \K B^* + C_7 \K \cdot \EHC
\end{equation}
where $C_1$ and $C_7$ are constants such that $C_1$ is sufficiently larger than 2, but smaller than all the constants $C_2, \ldots, C_6$ defined in Eq.~\eqref{eqn:Hpot}, and $C_7$ is a constant sufficiently larger than $C_2, \ldots, C_6$. We refer the reader to Table~\ref{table:potential} for a summary of the definitions of all variables defined in this section.

\begin{table*}
\caption{Definition of terms related to the potential function $\phi$.}
\label{table:potential}
\begin{center}
\begin{tabular}{ cl }
 \hline\hline
\textbf{Parameter} & \textbf{Definition} \\
\hline
 $T_{u,v}$ & Transcript of communication between $u$ and $v$ according to $u$ \\
 $G_{u,v}$ & Size of longest common prefix of $T_{u,v}$ and $T_{v,u}$ (in chunks)\\
 $B_{u,v}$ & $\max(|T_{u,v}|, |T_{v,u}|) - G_{u,v}$ \\ 
 $G^*$ & $\min_{(u,v) \in E}G_{u, v}$ \\ 
 $H^*$ & $\max_{u}\max_{v \in N(u)} |T_{u,v}|$ \\
 $B^*$ & $H^* - G^*$\\  
$\varphi_{u,v}$ & Meeting points potential between $u$ and $v$\\
$\EHC$ & Number of errors and hash collisions that have occurred overall \\
$\phi$ & Overall potential in network \\
\hline\hline
\end{tabular}
\end{center}
\end{table*}

\begin{remark}[Remark on Notation]
For any variable $\var$ that represents the state of some party in \Alg, including all the ones in Table~\ref{table:potential}, we let $\var(i)$ denote the value of the variable $\var$ at the beginning of iteration $i$. For example, $T_{u,v}(10)$ denotes the value of the partial transcript $T_{u,v}$ at the very start of the tenth iteration of 
\Alg.
\end{remark}

\subsection{The meeting-points mechanism and potential $\varphi_{u,v}$}\label{subsec:MP}

In what follows, we briefly recall the meeting-points mechanism and why we use it. We defer the formal definition of $\varphi_{u,v}$ and all the proofs regarding it to Appendix \ref{sec:meeting-points}. 

If two adjacent parties $u$ and $v$ have $T_{u,v} \neq T_{v,u}$ (or equivalently, $B_{u,v} > 0$), then they should not simulate further with each other without repairing the differences in their transcripts. If $u$ and $v$ knew which of them needs to roll back and by how much, they could simply roll back the simulated chunks until $T_{u,v} = T_{v,u}$, at which point they can continue the simulation. However, they do not know this information. Furthermore, they cannot afford to communicate $|T_{u,v}|$ or $|T_{v,u}|$, since these numbers potentially require $\log |\Pi|$ bits to communicate.

This problem is solved via the ``meeting-points'' mechanism~\cite{haeupler14} which is designed to roll back $T_{u,v}$ and $T_{v,u}$ to a point where $T_{u,v} = T_{v,u}$, while only requiring $O(B_{u,v})$ exchanges of hashes between parties~$u$ and~$v$, and guaranteeing that (in the absence of error) neither $u$ nor $v$ truncate their transcript too much. That is, $u$ (resp.~$v$) truncates $T_{u,v}$ (resp.~$T_{v,u}$) by at most $2B_{u,v}$~chunks. 
While errors and hash collisions can mess up this guarantee, each error or hash collision causes only a bounded amount of damage.  
Since the adversary's allowed error rate is sufficiently small, the simulation overcomes this damage with high probability.

As mentioned above, our analysis of the meeting-points mechanism  essentially follows that of Haeupler~\cite{haeupler14} after adapting it to our construction, where the meeting-points mechanism is interleaved over several iterations, rather than performed all at once.
Specifically, for each link $(u,v)\in E$ we define a ``meeting-points potential'' term~$\varphi_{u,v}$ that approximately measures the number of hash exchanges it will require for $u$ and $v$ to roll back $T_{u,v}$ and $T_{v,u}$ to a common point. 
While our analysis of how $\varphi_{u,v}$ changes in the meeting-points phase naturally repeats the analysis of~\cite{haeupler14}, $\varphi_{u,v}$ can also change during the other phases of the protocol, especially when noise is present. Our analysis bounds the change in~$\varphi_{u,v}$ in all the phases as a function of the errors and hash collisions that occur throughout the iteration. This allows us to bound the change in the overall potential~$\phi$. We bound the changes in~$\varphi_{u,v}$ in the Flag Passing, Rewind, and Simulation phases in Claim~\ref{claim:simple-Hpot-claim}. The changes in~$\varphi_{u,v}$ in the Meeting Points phase are addressed in Lemma~\ref{lem:potential-rises-overall-1} (analogous to Lemma 7.4 in~\cite{haeupler14}) and Proposition~\ref{prop:status-simulate}, and are combined to establish how the potential $\phi$ changes in the Meeting Points phase (Lemma~\ref{lem:network-potential-rises}).

We defer the formal definition of the meeting-points mechanism and the proofs of the relevant properties to Appendix~\ref{sec:meeting-points}.

\subsection{Bounding the potential increase and communication per iteration} %
\label{sec:potential-increase}
In this section we prove the following technical lemma that
says that the potential $\phi$ (Eq.~\eqref{eqn:potential}) increases in each iteration by at least~$\K$.
Furthermore, the amount of communication performed during a single iteration
can be bounded by roughly $\K$~times the amount of links (i.e., pairs of parties) that suffer from channel-noise during this iteration, or links  that experienced an event of hash-collision during this iteration.
\begin{lemma}
\label{lem:phi-rises-enough}
Fix any iteration of Algorithm~\ref{alg:robust-protocol} and let $\ell$ be the number of links with errors or hash collisions on them during this iteration. Then,
\begin{enumerate}
\item The potential $\phi$ increases by at least $\K$ in this iteration.
\item The amount of communication in the entire network during this iteration ($\CC$) satisfies
\[
\CC \leq \alpha(1 + \ell)\K,
\]
where $\alpha$ is a sufficiently large constant.
\end{enumerate}
\end{lemma}

The next sections are devoted to proving the above lemma. Let us begin by giving a high-level overview of the proof.
\subsubsection{Proof Overview}

We proceed to prove the lemma in two conceptual steps.

\begin{enumerate}
\item First, in Section~\ref{sec:potential-rises-no-err}, we consider iterations that have
no errors or hash collisions.

We first establish that the communication in this case is at most $O(\K)$.
To this end, we first argue that the communication in the meeting-points, flag-passing, and rewind
phases is \emph{always} bounded by $O(\K)$ (Proposition \ref{prop:easy-comm-bound}), regardless of errors
committed by the adversary.  Therefore, it suffices to bound the communication in the simulation phase.
If every party is simulating the same chunk, then the
communication is easily bounded by $O(\K)$. However, if the parties are simulating many different chunks, then the
communication could be much larger. This is where the flag-passing phase is useful:
if there are no errors, then the flags will prevent all parties from simulating when two parties are at different chunks.

We next establish that the potential increases by at least~$\K$, as follows.
If the parties simulate, then since there are no errors or hash collisions, $\sum G_{u,v}$ increase by $\K$, and none of the other terms change. If the parties do not simulate, then either some adjacent parties did not pass their consistency check, in which case $\varphi_{u,v}$ decreases by $\Omega(1)$ (Lemma \ref{lem:potential-rises-overall-1}) and none of the other terms change, or some parties rewind, in which case $B^*$ decreases and none of the other terms change.

\item Next, in Section~\ref{sec:potential-rises-err-mp}, we consider iterations that have errors or hash collisions.

We first argue that errors and hash collisions increase $\phi$ by at least $\K$.
To this end, note that errors may cause some terms of $\phi$ to decrease,
but this is compensated for by the accompanying increase in $\EHC$, and since $C_7$ is set to be large enough, even though some of the terms decrease, overall the potential increases by at least~$\K$.

We would then like to argue that the communication increases by at most $O(\K)$, though unfortunately, this claim is false.  The communication in an iteration can actually greatly exceed $O(\K)$, though we show that in these cases, there were many errors or hash collisions in the iteration.
Specifically, we argue that each error or hash collision {\em individually} does not cause too much extra communication.  This is formalized in Lemma~\ref{lem:errs-in-mp-and-fp}.

\end{enumerate}

First, we prove a simple proposition, which says that the communication in the meeting-points, flag-passing, and rewind phases is bounded. This reduces bounding the overall communication
in an iteration to bounding the communication in the corresponding simulation phase.

\begin{proposition}
\label{prop:easy-comm-bound}
The communication during the flag-passing and rewind phases is $O(m)$ in total, and the communication in the meeting-points phase is $O(\K)$, regardless of errors or hash collisions in the iteration. \end{proposition}
\begin{proof}
In the meeting-points phase, each adjacent pair of parties exchange hashes of their transcripts (see Algorithm \ref{alg:meeting-points}), where the output length of the hash functions is $\Theta(\K / m)$. Hence, there is $O(\K)$ communication in the meeting-points phase. \\
The communication pattern in the flag-passing phase is deterministic and consists of two messages per link of a the spanning tree~$\cal{T}$, hence it is upper bounded by $O(n)=O(m)$.
Finally, each link can have at most one valid ``rewind'' message in the rewind phase (note that messages that are inserted do not count towards our communication bound).
\end{proof}

\subsubsection{Iterations with no errors or hash collisions}
\label{sec:potential-rises-no-err}
\begin{lemma}
\label{lem:no-err:comm}
Suppose that there are no errors or hash collisions in a single iteration of Algorithm~\ref{alg:robust-protocol}. 
Then the overall  communication in the network is~$O(\K)$.
\end{lemma}
\begin{proof}
By Proposition \ref{prop:easy-comm-bound}, the communication in all the phases except of the simulation phase are bounded by~$O(\K)$, and we are left to bound the communication in the simulation phases.

In the simulation phase, each party either sends $\bot$ or simulates a specific chunk. Say that $v$ simulates chunk number $n_v$ with all its neighbors if it didn't send~$\bot$.
Each chunk contains at most $5\K$ bits of communication, hence, the total amount of communication in the simulation phase is
bounded by $5\K$ times the number of distinct chunk numbers being simulated in the network.
In other words, it is bounded by $5\K \cdot | \{ n_v \mid v \in V\}|$, up to additional $2m$ $\bot$ ``messages'' (which in our case are merely $2m$ bits).

Therefore, to finish the proof it remains to argue that if there are no errors or hash collisions then $ | \{ n_v \mid v \in V\}| \leq 1$.  
We consider several different cases according to the state of the network at the beginning of the iteration,
specifically, whether the parties have set $netCorrect=1$ or not.

\paragraph{Case 1: At the end of the flag-passing phase, $netCorrect_u=1$ for every party $u$.}

Since there were no errors or hash collisions, the fact that $netCorrect_u=1$ means that each party~$u$ had $status_u=1$ before the flag-passing phase. This follows since by the definition of the flag-passing phase, for every party $v\in V$, $netCorrect_v=\bigwedge_{u\in V} status_u$.   Note that, since we assume that none of the parties have $netCorrect_u = 0$, and we assume no errors, the are no $\perp$ symbols sent in the simulation phase.

The fact that each party $u$ has $status_u=1$
implies that for all $v, w \in N(u)$,  it holds that $|T_{u, v}| = |T_{u, w}|$.
Further, for any $v \in N(u)$, $T_{u,v}=T_{v,u}$, or otherwise the hashes would indicate a mismatch and the parties would have set $status=0$.
Putting these two facts together, we get that $G^* = H^*$ and hence $B^* = 0$, which implies that indeed $ | \{ n_v \mid v \in V\}|=1$, as desired.

\paragraph{Case 2: At the end of the flag-passing phase, some party~$u$ has $netCorrect_u=0$.}
Since $netCorrect_u=0$ for some party~$u$, there must be some party $v$ such that $status_v = 0$, and hence we have that $netCorrect_u=0$ for all $u \in V$.
Since $netCorrect_u = 0$ for all parties $u$, we know that none of the parties will simulate (they will only send $\perp$s) and hence the overall communication in the iteration will be $2m=O(\K)$.

\end{proof}

We next show that the potential increases by at least~$K$ in any such iteration.

\begin{lemma}
\label{lem:no-err:potential}
Suppose that there are no errors or hash collisions in a single iteration of Algorithm~\ref{alg:robust-protocol}. 
Then the potential $\phi$ increases by at least~$\K$ during this iteration.
\end{lemma}

\begin{proof}
We consider the status of the network at the iteration according to the next three cases.

\paragraph{Case 1: At the end of the flag-passing phase, $netCorrect_u=1$ for every party $u$.}

Recall that since there were no errors or hash collisions, the fact that $netCorrect_u=1$ means that each party~$u$ had $status_u=1$ before the flag-passing phase. 
This in turn implies that 
for all $v, w \in N(u)$,  it holds that $|T_{u, v}| = |T_{u, w}|$.
Further, for any $v \in N(u)$, $T_{u,v}=T_{v,u}$, or otherwise the hashes would indicate a mismatch and the parties would have set $status=0$. Consequently, we have $G^* = H^*$ and $B^* = 0$.
The fact that $netCorrect_u=1$ for every party $u$, together with the fact that $B^*=0$, implies that all parties simulate the same chunk, and the absence of errors in the communication
implies that this simulation is done correctly.
Hence, each $T_{u,v}$ is extended correctly according to~$\Pi$.  This in turn implies that
$G_{u,v}$ increases for each $(u,v) \in E$, which causes $\phi$ to increase by~$\K$.

Next, we argue that none of the other terms of $\phi$ decrease.  We first argue that $B^*$ remains zero at the end of the iteration.
To this end, note that since all parties simulate one chunk in each of their pairwise transcripts, we still have the property that $|T_{u_1, v_1}| = |T_{u_2, v_2}|$ for all $(u_1, v_1) \in E$ and $(u_2, v_2) \in E$ after the simulation phase. Since there were no errors, we also have that $T_{u,v} = T_{v,u}$ for all $(u,v) \in E$. As noted before, this gives us that $B^* = 0$ after the simulation phase, and since there are no errors it remains zero after the rewind phase as well.

It remains to argue that $\varphi_{u,v}$ does not increase for any $(u,v) \in E$. 
By Proposition \ref{prop:status-simulate} we know that
 $\varphi_{u,v}$ does not increase in the meeting-points phase. Furthermore, it does not increase in the flag-passing, simulation or rewind phases either, by Claim \ref{claim:simple-Hpot-claim}.

Putting this all together, we have that each $G_{u,v}$ increases by one, $B^*$ does not change and $\varphi_{u,v}$ does not increase, which implies that the potential $\phi$ increases by at least~$\K$ overall, as desired.

\paragraph{Case 2: Some party $u$ has a neighbor $v \in N(u)$ s.t. $status_{u,v} = \text{``meeting points''}$ after the meeting-points phase.}

Since $status_{u,v} = \text{``meeting points''}$, $u$ has $status_u = 0$ after the meeting-points phase, and therefore, given the lack of errors, each party $x\in V$, will set $netCorrect_x=0$ after the flag-passing phase. 
Note that since none of the parties are simulating the next chunk during the simulation phase, it follows that $\phi$ does not change in the simulation phase.
Next note that the potential increases by at least $5K$ during the meeting-points phase (Lemma \ref{lem:network-potential-rises}).
Since the potential does not change during the flag-passing phase, it remains to argue that the potential function does not decrease by much during the rewind phase.

In the rewind phase, we may have parties that send rewinds. Even though these rewinds seem to take us in the right direction, they may cause a small decrease in some terms of the potential.
However, we argue that in the rewind phase the potential decreases by at most $\K$, and thus in total, $\phi$ increases by at least $5\K - \K = 4\K$, as desired.

To this end, first note that since we limit the number of truncations per link to at most one, it follows that $(\K / m)\sum_{(x,y) \in E} G_{x,y}$ can decrease by at most $\K$.
It remains to argue that $B^*$ and $\{\varphi_{x,y}\}_{(x,y)\in E}$ do not increase in the rewind phase.

The fact that $\varphi_{x,y}$ does not increase follows from Claim \ref{claim:simple-Hpot-claim}.  To argue that $B^*$ does not increase, note that since no party simulates in this iteration, $H^*$ does not increase. We claim a party~$x$ will never truncate a transcript $T_{x,y}$ such that $|T_{x,y}| = G^*$. Clearly $x$ will not send a rewind message to~$y$, since there is no $y^*$ such that $|T_{x,y^*}| < |T_{x,y}|$ by the definition of $G^*$. We also claim that $y$ will not send a rewind message to~$x$. If $|T_{y,x}| = |T_{x,y}| = G^*$, then this follows because there is no $x^*$ such that $|T_{y,x^*}| < |T_{y,x}|$. Otherwise if $|T_{y,x}| \neq |T_{x,y}|$, then since there were no hash collisions or errors in the meeting-points phase we conclude that $status_{y,x} = \text{``meeting points''}$. Therefore $y$ will not send a rewind message to $x$.

\paragraph{Case 3: At end of the meeting-points phase, $status_{u,v} = \text{``simulate''}$ for all $(u,v) \in E$, yet at the end of the flag-passing phase, some party~$u$ has $netCorrect_u=0$.}
Again, since $netCorrect_u=0$ for some party~$u$, 
there must be some party $v$ such that $status_v = 0$, and hence we have that $netCorrect_u=0$ for all $u \in V$.
No party simulates on the next simulation phase (they will only send $\perp$s), hence,
the potential does not change during the simulation phase.  
In addition,
the potential does not decrease in the meeting-points phase (Lemma \ref{lem:network-potential-rises}).  
Furthermore, the potential remains unchanged during the flag-passing phase.
Therefore, all that remains to show is that $\phi$ increases in the rewind phase by at least~$\K$.

  Recall that by Claim \ref{claim:simple-Hpot-claim}, $\varphi_{u,v}$ will not increase during the rewind phase for any $(u,v)$. Furthermore, $\sum G_{u,v}$ can decrease by at most $m$ in the rewind phase, which means $(\K/m)\sum G_{u,v}$ decreases by at most $\K$. Therefore, it suffices to show that $B^*$ decreases by 1, and therefore that $\phi$ increases by $C_1 \K - \K \geq \K$, since $C_1 \geq 2$.

To this end, we first show that $G^*$ does not decrease, and then show that $H^*$ decreases by~1.
For the former, note that a party $u$ will never issue a rewind to a party $v$ for which $|T_{u,v}| = G^*$, since this would mean that there is some party $w \in N(u)$ such that $|T_{u,w}| < G^*$, which contradicts the definition of $G^*$. Therefore, $G^*$ does not decrease.

To argue that $H^*$ decreases by~$1$, fix a party $u$ and a neighbor $v$ such that $|T_{u,v}| = H^*$. We argue that during the rewind phase party~$u$ will rewind this transcript by one chunk. The proof goes by induction on the distance between $u$ and a party whose $minChunk \ne H^*$.
To this end, let $S^* = \{v: \exists w \in N(v) : |T_{v,w}| < H^* \}$ be the set of parties that have some transcript below chunk~$H^*$. 
\begin{claim}
$S^*$ is non-empty.
\end{claim}
\begin{proof}
Indeed, given that all parties have $netCorrect_q=0$ after the flag-passing phase, despite all pairs $(q_1, q_2) \in E$ having $status_{q_1,q_2} = \text{``simulate''}$, it must hold that some party sees an inconsistency in the lengths of the simulated transcript with two of its neighbors (Line~\ref{step:minChunkCheck}). Namely, for some party~$q^*$, there are neighbours $w,w'$ such that $|T_{q^*,w}| \ne |T_{q^*,w'}|$.
It follows that $|T_{q_1,q_2}| = H^*$ cannot hold for all $(q_1,q_2)\in E$.
\end{proof}

Let $d(q,S^*)$ denote the shortest distance in the graph~$G$ between a party~$q$ and some party in~$S^*$, where $d(q,S^*) = 0$ if and only if $u \in S^*$. 
\begin{claim}
Party $u$ truncates $T_{u,v}$ after at most $d(u,S^*)+1$ rounds of the rewind phase.
\end{claim}
\begin{proof}
Note that if $d(u,S^*)=0$, i.e.,
$u \in S^*$, then $minChunk_u < H^*$ by the definition of $S^*$. So, in the rewind phase, $u$ will truncate $T_{u,v}$ (Line~\ref{step:truncate}). 

Next we claim that if $d(u,S^*) = j$ for some $j>0$ at the beginning of some round~$r$ in the rewind phase, then in the beginning of round~$r+1$ it holds that $d(u,S^*) = j-1$. 
Denote $u$ as $a_0$. Let $a_0, a_1, a_2, \ldots, a_j \in S^*$ be the vertices in a shortest path from $a_0$ to $S^*$.
Note that for any two consecutive parties along this path, $|T_{a_i, a_{i-1}}| = |T_{a_{i-1}, a_i}| = H^*$. This is true since $a_0,a_1,\ldots, a_{j-1}\notin S^*$, and since  and $status_{a_j, a_{j-1}} = status_{a_{j-1}, a_{j}} = \text{``simulate''}$, which means any two parties are consistent with their transcripts.

Since $a_{j}\in S^*$ we have that $minChunk_{a_j} < H^*$ and it follows that 
in round~$r$, party~$a_j$ sends a rewind message to~$a_{j-1}$ (Line~\ref{step:truncate}).
We stress that no rewind message has yet been sent on the link $(a_j, a_{j-1})$.
Indeed, if this were not the case, then we would have gotten $|T_{a_j, a_{j-1}}| < H^*$ already in a prior round where the rewind message took place. But this contradicts $a_{j-1} \notin S^*$ in the beginning of round~$r$.
Hence, by the end of round~$r$ we have that $|T_{a_j, a_{j-1}}| = |T_{a_{j-1}, a_j}| = H^*-1$
and thus $a_{j-1}\in S^*$. This means that $d(u, S^*)=j-1$ at the beginning of round~$r+1$.

By employing the same argument inductively we get that,
if at the beginning of the rewind phase $d(u,S^*)=j$, then after
$j$~rounds we have that $u\in S^*$, and after the $(j+1)-th$ round, party 
$u$ has truncated $T_{u,v}$ by at least one chunk, as needed.
\end{proof}
Showing that $H^*$ has decreased by 1 is now straightforward.
For any party $u$ we note that $d(u,S^*)$ can never exceed $n-1$, which is an upper bound of the diameter of~$G$. Since the rewind phase consists of $n$ rounds, after its $(n-1)$-th round, all parties are in~$S^*$, and by the end of the $n$-round of the rewind phase, all pairwise transcripts are of length at most $H^*-1$. This completes the proof of Lemma~\ref{lem:no-err:potential}.
\end{proof}

\subsubsection{Iterations with errors and hash collisions}
\label{sec:potential-rises-err-mp}
\begin{lemma}
\label{lem:errs-in-mp-and-fp}
 Let $\ell$ be the number of links that experienced either errors or hash collisions during a given iteration, and assume $\ell \geq 1$.
 Then the increase in the potential $\phi$ in this iteration is at least $\Omega(C_7 \ell \K)$,
 and the amount of communication in this iteration is at most~$O(\ell\K)$.
\end{lemma}
While $C_7$ is a constant, we include it in the $\Omega$ to indicate that, by making $C_7$ larger,
we can overwhelm the constant hidden in the $\Omega$. This will be useful in proving Lemma~\ref{lem:phi-rises-enough}.

\begin{proof}%
Let $\ell_1$ denote the number of links with errors and hash collisions in the meeting-points phase, let $\ell_2$ denote the number of links with errors in the flag-passing phase, let $\ell_3$ denote the number of links with errors in the simulation phase, and let $\ell_4$ denote the number of links with errors in the rewind phase. Then $\ell \leq \ell_1+\ell_2+\ell_3+\ell_4$.

Let us begin with bounding the communication in this iteration.
By Proposition \ref{prop:easy-comm-bound}, the communication in all the phases except for the simulation phase are bounded by~$O(\K)$, and we are left to bound the communication in the simulation phase.
As explained in the proof of Lemma~\ref{lem:no-err:comm}, the amount of communication in the simulation phase is bounded by $5\K$ times the number of distinct chunk numbers being simulated in the network (plus $2m$ for $\perp$s). We show that this number is proportional to the number of links that experienced errors or hash collisions during that same iteration.

Let $\cal{T}$ be the spanning tree used for the flag passing in Algorithm~\ref{alg:flag-passing}.
Consider the subgraph $\mathcal{H}$ of~$\cal{T}$ induced by only keeping an edge $(u,v) \in \mathcal{T}$ if $netCorrect_u = netCorrect_v = 1$ and $|T_{u,v}| = |T_{v,u}|$.
Recall that $netCorrect_u = 1$ implies that $status_u = 1$,
and thus for any $u$ with $netCorrect_u=1$ we know that $|T_{u,v}| = |T_{u,w}|$ for all $w\in\N(u)$ and in particular for all $w$ such that $(u,w) \in \mathcal{H}$. 
By a straightforward induction, one can argue that for any pair of parties~$u$ and~$x$ that are in the same connected component of $\mathcal{H}$, it holds that $|T_{u,v}| = |T_{x,y}|$ for any $v,y$ s.t. $(u,v) \in \mathcal{H}$ and $(x,y) \in \mathcal{H}$. Hence, in any single connected component of $\mathcal{H}$, at most one chunk of $\Pi$ is being simulated. 
Note that there might be components in $\mathcal{H}$ with \emph{no} chunk being simulated.  Each such connected component consists of a single isolated variable~$u$ such that  $netCorrect_u=0$.

\begin{claim}
\label{claim:communication-bound}
Let $S$ denote the set of connected components in $\mathcal{H}$ such that $netCorrect_u = 1$, and let~$s=|S|$ be the number of components in~$S$.  
Then,
\[
s -1 \le \ell_1 + \ell_2 .
\]
\end{claim}

\begin{proof}
We claim that there are at least $s-1$ edges $(u,v)$ in $\mathcal{T} \setminus \mathcal{H}$ such that $v$ is in a component in $S$ and $\ell(u) < \ell(v)$, where recall that $\ell(u)$ is defined to be the distance of~$u$ from $\rho$, which is the root of~$\mathcal{T}$, plus~$1$. Towards seeing this, note that each connected component in $S$ is a subtree of $\mathcal{T}$, and since they are disjoint,  at least $s-1$ many of them do not have $\rho$ as the root of the subtree. Let $v$ be such a root, and let $u$ be its parent in $\mathcal{T}$. This satisfies the desired conditions.

Fix such an edge $(u,v)$. We argue that there was an error or hash collision on the link $(u,v)$ in either the meeting-points or flag-passing phases, establishing the claim.
Since $(u,v)$ is an edge in $\mathcal{T}\setminus \mathcal{H}$, we know that either $netCorrect_u = 0$, $netCorrect_v = 0$, or $|T_{u,v}| \neq |T_{v,u}|$; otherwise, $(u,v)$ would have been in~$\mathcal{H}$. 
However, we know that $netCorrect_v = 1$, since $v$ is in a connected component in $S$ and by the definition of~$S$. 
Hence, it can either be that $netCorrect_u = 0$ or that $|T_{u,v}| \neq |T_{v,u}|$.

If $|T_{u,v}| \neq |T_{v,u}|$,
then there must have been an error or hash collision in the meeting-points phase, otherwise we would have $status_v=0$ implying $netCorrect_v = 0$, which is a contradiction. 
If, on the other hand, $netCorrect_u = 0$ holds, then there must have been an error in the downward part of the flag-passing phase, since $netCorrect_v = 1$, so clearly $v$ did not correctly receive the flag that $u$ sent.
\end{proof}

As argued above, the communication during the simulation phase is bounded by $s \cdot 5\K + 2m$:
 each connected component in~$S$ jointly simulates a single chunk, and components outside of~$S$ (which consists of a single party $u$ such that $netCorrect_u = 0$) do not simulate; the $2m$ term comes from potential $\bot$ messages.
The above claim implies that $s\le \ell_1+\ell_2 +1 =O(\ell)$, leading to communication of $O(\ell \K)$ in the simulation phase. Since all the other phases have communication $O(\K)$, the communication in the entire iteration is as claimed.

\medskip
To finish the proof of the Lemma it remains to bound the increase in the potential~$\phi$. 
Consider the various phases in the iteration, and the terms of~$\phi$ given by Eq.~\eqref{eqn:potential}.
\begin{itemize}
\item
\textbf{Meeting Points:}
Lemma~\ref{lem:network-potential-rises} guarantees that the potential $\phi$ goes up by at least $5c \K + 0.4C_7 \ell_1  \K$, where $c$ is the number of pairs of parties $(u,v)$ such that $(u,v) \in E$ and $status_{u,v}$ or $status_{v,u}$ is $\text{``meeting points''}$ at the end of the Meeting Points phase.

\item \textbf{Flag Passing:} No direct potential change happens in this phase, other than any increase in potential caused by an error induced by the adversary. So the potential in this phase increases by at least $C_7 \ell_2 \K$.

\item \textbf{Simulation:}
The term $\sum G_{u,v}$ cannot decrease in the simulation phase, since no transcript is being truncated in this phase.
Claim \ref{claim:simple-Hpot-claim} establishes that $\varphi_{u,v}$ increases by $C_3$ in the simulation phase only if there was an error on the link $(u,v)$ somewhere in this iteration. Otherwise, it does not increase. Hence, the term $-\sum \K \cdot \varphi_{u,v}$ decreases by at most $\K C_3 \ell$.
The third term,$-C_1 \K B^*$, decreases by at most $\K C_1$.
Indeed, $B^* = H^* - G^*$ increases by at most 1 in the simulation phase, since $H^*$ can increase by 1 but $G^*$ cannot decrease, since no party truncates in this phase.
The fourth term, $C_7 \K \cdot \EHC$, increases by at least $C_7 \ell_3 \K$.
Thus in total, the potential function increases during the simulation phase by at least
\(  - C_3 \ell \K -  C_1 \K + C_7 \ell_3 \K\).
\item \textbf{Rewind:}
The term $(\K / m)\sum G_{u,v}$ can decrease in the rewind phase by at most~$\K$, since each party rewinds a transcript at most one chunk. The second term $-\sum \K \cdot \varphi_{u,v}$ decreases by at most $\K C_3 \ell$, again, by Claim \ref{claim:simple-Hpot-claim}.
The third term $-C_1 \K B^*$ decreases by at most $C_1\K$ since $B^*$ increases by at most 1 in the rewind phase: $G^*$ can decrease by at most one (since no party rewinds more than one chunk) and $H^*$ cannot increase.
The fourth term $C_7 \K \EHC$  increases by $C_7 \ell_4 \K$.
\end{itemize}

Putting it all together, we get that in the entire iteration $\phi$ increases by at least
\begin{align*}
&5c \K + 0.4C_7 \ell_1  \K   %
+ C_7 \ell_2 \K   %
- C_3 \ell \K -  C_1 \K + C_7 \ell_3 \K   \\
& \phantom{5c \K} %
-\K - C_3 \ell \K -  C_1 \K + C_7 \ell_4 \K  %
\\
&\ge 5c \K +
0.4 C_7(\ell_1 + \ell_2+\ell_3+\ell_4)\K -
2 C_3 \ell \K - 2C_1 \K - \K \\
&\ge 5c \K + (0.4 C_7 - 2C_3)\ell \K - (2C_1 +1)\K
\end{align*}
where the final inequality follows from the fact that $\ell_1 + \ell_2 + \ell_3+\ell_4 \ge \ell$. Since $\ell \geq 1$ (by our assumption), we can take $C_7$ to be sufficiently large compared to $C_3$ and $C_1$, and get that the change in $\phi$ is $\Omega(C_7 \ell \K)$, as desired.
\end{proof}

\subsubsection{Putting it all together}
We can finally complete the proof of Lemma~\ref{lem:phi-rises-enough}.
\begin{proof}[Proof of Lemma~\ref{lem:phi-rises-enough}]
We first recall that Lemma~\ref{lem:no-err:potential} establishes that, in the absence of errors and hash collisions, the potential increases by~$\K$ and Lemma~\ref{lem:no-err:comm} bounds the total communication by~$O(\K)$. Now assume that there is at least one error or hash collision in the iteration, so $\ell \geq 1$. Then Lemma~\ref{lem:errs-in-mp-and-fp} shows that the potential increases by $\Omega(C_7 \ell \K)$, and that the communication in the iteration is at most $O(\ell \K)$.
By taking $C_7$ to be sufficiently large, we can see that $\phi$ increases by at least~$\K$ while the communication is bounded by
$O((\ell+1) K)$, as required.
\end{proof}

\changelocaltocdepth{2}   %
\subsection{Bounding  hash collisions and communication}
\label{sec:bounding-collisions}
In this section we prove that the number of hash collisions throughout the entire simulation is bounded by $O(\eps |\Pi|)$ with high probability, where $|\Pi|$ is the number of chunks in the original, noiseless protocol.

The main lemma we prove in this subsection is the following. We remind the reader that the number of channel errors an oblivious adversary makes is deterministic, and hence does not depend on any randomness used in Algorithm~\ref{alg:robust-protocol} (Remark~\ref{rem:additive}).
\begin{lemma}
\label{lem:bounding-collisions-main}
Let $\eps > 0$ be a sufficiently small number, independent of $m$ and $|\Pi|$.
Suppose we have a hash function $h$ in Algorithm~\ref{alg:robust-protocol} with hash collision~$p$, where
$p < \frac{1}{10C_6}$.
Suppose there is an oblivious adversary, and 
let $\Err$ denote the number of channel errors the adversary makes. Let $\CC$ denote the total communication in the entire execution of Algorithm \ref{alg:robust-protocol}, and $\EHC$ denote the joint number of errors and hash collisions during that execution. Let $k$ be an integer such that $k \geq 10C_6$ and $k = O(1 / \eps)$. Then, with probability $1 - p^{\Omega(k \eps |\Pi|)}$, the event   
\begin{align*}
&\big(
\CC \leq 200\alpha |\Pi| \K 
\quad\text{ and } \quad
\EHC \leq (k+1) \cdot (200 \alpha \eps) |\Pi| \big )  
	\\ & \text{ or } \Err > \frac{\eps}{\K} \CC
\end{align*}
holds,
where $\alpha$ is the constant multiplying the communication complexity of an iteration in Lemma~\ref{lem:phi-rises-enough}.
\end{lemma}
\ran{reworded similar to 4.1}

\subsubsection{Overview and A Note on Oblivious Adversaries}
\label{sec:bound-hash-overview}
We will use hash functions with constant collision probability $p = 2^{-\Theta(\K/m)}$, which is far higher than the adversarial error rate of $\varepsilon / \K$ when $\K = m$. Despite this, we can still bound the number of hash collisions that occur in the protocol overall by $O(\eps |\Pi|)$. This will follow from the observation that hash collisions are \emph{one-sided}---they can only happen when the transcripts being hashed are different. Since the meeting points protocol lets the parties correct their errors in relatively few steps, there will be few opportunities for hash collisions. 
A similar approach is taken in~\cite{haeupler14} in the two-party setting.

Now we give more details. Fix two parties $u$ and $v$. Note that $\varphi_{u,v}$ roughly measures how many hashes must be passed between the two parties to get back to a consistent transcript, with the property that $u$ and $v$ have a consistent transcript when $\varphi_{u,v} = 0$ (Proposition \ref{prop:phi-nonnegative}). The main idea is that the potential function $\varphi_{u,v}$ can only increase by at most a constant during any single exchange of meeting points, even in the presence of errors or hash collisions. Furthermore, in the absence of errors or hash collisions, it decreases by some constant. Finally, if $\varphi_{u,v} = 0$, then it can only increase above 0 if an error is introduced between $u$ and $v$, because hash collisions do not occur when the transcripts match.

The main approach of this section is to argue that $\varphi_{u,v}$ should not be nonzero too often. For intuition, suppose the adversary starts by making some small number of errors, which makes $\varphi_{u,v}$ equal to some number $N$. Then the number of iterations in which $\varphi_{u,v}$ is nonzero will be small as long as hash collisions happen infrequently enough that the resulting increase in $\varphi_{u,v}$ does not outweigh the decrease of $\varphi_{u,v}$ in a typical iteration, where a hash collision does not happen. This will follow from the independence of hash collisions across iterations and links of the network. The independence of hash collisions is a corollary of having an oblivious adversary, which we now describe further.

\gnote{Read here 2021}\ran{this entire part looks trivial to me. But, oh well. Another page to the 50.}
An oblivious adversary places their errors before seeing the execution of the protocol. What will specifically be useful to us is that they place their errors without knowing the randomness~$S$ used to seed the hash function\footnote{Indeed, simply hiding the randomness from the adversary is enough to get all of our results.}. Hence, the transcripts at the beginning of iteration $i$ are a function of the parties' inputs, the adversary's errors through iteration $i-1$, and the randomness used by the parties in iterations $1, \ldots, i-1$. The random seed used by pair of parties $(u,v) \in E$ and each iteration $i \in [100|\Pi|]$ is independent of all of these things; so the seed is selected independently from the transcript it is used to hash, and the events of hash collisions are independent. We briefly formalize this now. 

\begin{observation}
\label{obs:oblivious}
Suppose the adversary is oblivious. Fix any iteration $i$, and any link $(u,v) \in E$.
Let $S_{i,u,v}$ denote the random seed sampled for the link $(u,v)$ in iteration $i$, conditioned on an arbitrary setting of the transcripts at the beginning of the iteration $T_{u,v}(i)$ and $T_{v,u}(i)$. The (conditional) distribution of $S_{i,u,v}$ is uniform. 
\end{observation}
\begin{proof}
Conditioned on any fixing of all the random strings used in iterations $1, \ldots, i-1$, the seed $S_{i,u,v}$ is a uniformly random bit string; this follows because $S_{i,u,v}$ is sampled independently of all the previous seeds.

The seeds in previous iterations also uniquely determine $T_{u,v}(i)$ and $T_{v,u}(i)$, as fixing these strings eliminates all randomness of the protocol in iterations $1, \ldots, i-1$. This is where we use the fact that the adversary is oblivious; if they were not oblivious (that is, they have access to all the seeds that will be used throughout the protocol before it even happens; see~\cite{GKR-2} for details), they could choose their errors as a function of the random seed $S_{i,u,v}$, and so $T_{u,v}(i)$ and $T_{v,u}(i)$ could have a dependence on $S_{i,u,v}$, rather than just the seeds from iterations $\leq i-1$.

By summing over all settings to the randomness in iterations $1, \ldots, i-1$ that lead to the same settings for $T_{u,v}(i)$ and $T_{v,u}(i)$, we conclude the observation.
\end{proof}

This observation will let us bound the number of hash collision with high probability by using a Chernoff bound. We note that the only property of an oblivious adversary we require here is that the error pattern is independent of the seeds used throughout the protocol. Hence, the results in this section hold for \emph{any} oblivious adversary, regardless of whether the adversary is additive or fixing (see Section~\ref{sec:preliminaries} for the definitions of these adversaries).

\subsubsection{Proof of Lemma~\ref{lem:bounding-collisions-main}}
In analogy with the terminology of dangerous rounds from~\cite{haeupler14}, we define the notion of dangerous triples as follows.

\begin{definition}
Let $i$ be an iteration of \Alg, and let $u$ and $v$ be parties such that $(u,v) \in E$. Call the triple $(i,u,v)$ \emph{dangerous} if $B_{u,v} > 0$ at the beginning of iteration $i$.
\end{definition}

Now we state the lemma that we will prove in this section, which will be the main workhorse in proving Lemma~\ref{lem:bounding-collisions-main}.

\begin{lemma}
\label{lem:hash-collision-bound-obliv}
Let $\eps > 0$ be a sufficiently small number independent of $m$ and $|\Pi|$. Suppose an oblivious adversary, and
suppose the hash collision probability of $h$ is $p$ such that $p < \frac{1}{10C_6}$ in Algorithm \ref{alg:meeting-points}. Let $\Err$ denote the number of errors the adversary makes, $\CC$ denote the total communication in 
\Alg,
and let $D$ denote the number of dangerous triples $(i,u,v)$. Let $k$ be a an integer such that $k \geq 10C_6$ and $k = O(1 / \eps)$. Then with probability $1 - p^{\Omega(k \eps |\Pi|)}$, the event 
\begin{align*}
&\big (\CC \leq 200\alpha |\Pi| \K \quad \text{and}  \quad D \leq k \cdot (200 \alpha \eps) |\Pi| \big) \\ &\text {or }\Err > \frac{\eps}{\K} \CC
\end{align*}
holds, where $\alpha$ is the constant multiplying the communication complexity of an iteration in Lemma~\ref{lem:phi-rises-enough}.
\end{lemma}
Note that a hash collision can only occur between $u$ and $v$ in an iteration $i \in [100 |\Pi|]$ when $(i,u,v)$ is a dangerous triple, by definition. Hence, the proof of Lemma~\ref{lem:bounding-collisions-main} from Lemma~\ref{lem:hash-collision-bound-obliv} follows easily:
\begin{proof}[Proof of Lemma~\ref{lem:bounding-collisions-main}]
Suppose $\Err \leq \frac{\eps}{\K} \CC$, and recall that $\log(1/p) = \Theta(\K / m)$. By Lemma~\ref{lem:hash-collision-bound-obliv}, the communication of \Alg is upper bounded by $200\alpha |\Pi|\K$ with probability $1 - p^{\Omega(k \eps |\Pi|)}$. When this occurs and also $\Err \leq \frac{\eps}{\K} \CC$, this implies that the number of errors is at most $200 \alpha \eps |\Pi|$. Furthermore, the number of hash collisions is at most the number of dangerous triples, which is at most $k \cdot (200 \alpha \eps) |\Pi|$, again by Lemma~\ref{lem:hash-collision-bound-obliv}. This concludes the proof that $\EHC$ is bounded by $(k+1)(200 \alpha \eps) |\Pi|$.
\end{proof}

The first part of Lemma \ref{lem:hash-collision-bound-obliv} argues that the communication complexity $\CC$ is bounded with high probability. First we will argue that if the communication complexity is too large, then the number of dangerous triples must be very large with respect to the number of errors the adversary can introduce. Then, we establish that the number of dangerous triples can only be so large if the fraction of hash collisions in these triples is too large, which happens with low probability.

\begin{lemma}
\label{lem:cc-bound}
Consider a run of Algorithm~\ref{alg:robust-protocol}, with~$\Err$ errors.
Let~$D$ be the number of dangerous triples, and assume 
$\CC > 200\alpha |\Pi| \K$, where $\alpha$ is the constant multiplying the communication in Lemma \ref{lem:phi-rises-enough}. Additionally, suppose that $\Err \leq \frac{\eps}{\K} \CC$.

If $\Err > 0$, then $D \geq \beta \cdot \Err$, where $\beta \eqdef \frac{\CC}{3\alpha\K\Err} \geq \frac{1}{3\alpha \eps}$. If $\Err = 0$, then $D=0$ trivially.
\end{lemma}

\begin{proof}
The final statement follows easily - if there is never any error in the protocol, then there will never be a point at which any pair of transcripts are mismatched. 

Now we prove the statement when $\Err > 0$. We start by summing Lemma~\ref{lem:phi-rises-enough} over all the $100|\Pi|$ iterations $i$ of Algorithm~\ref{alg:robust-protocol} in order to bound the communication complexity of the entire protocol. 
Note that if we let $\ell(i)$ denote the number of errors and hash collisions in iteration $i$, then $\sum_{i\in[100|\Pi|]} \ell(i) \leq \EHC$, where $\EHC$ is the total number of errors and hash collisions experienced throughout the entire protocol. 
\begin{align*}
\CC & \le \sum_i  \alpha(1+\ell(i))\K \\
&\leq \alpha \K( 100|\Pi|   + \EHC ) \\
& \leq \alpha \K \left(100 |\Pi| + \frac{\CC \cdot \eps}{\K} + D\right)
\end{align*}
where the first inequality is Lemma~\ref{lem:phi-rises-enough}, and in the last inequality we use the fact that 
\[
\EHC \le \frac{\eps}{\K}\CC + D
\]
since $D$ upper bounds the number of hash collisions, and that the error rate is bounded by~$\eps / \K$. 
Rearranging, we get that
\begin{align*}
D &\geq \frac{(1 - \alpha \eps)\CC - 100 \alpha |\Pi| \K}{\alpha \K} \\
&\geq \frac{\CC}{3\alpha\K} \\
&= \beta \Err
\end{align*}
where the second inequality follows from the fact that $\CC > 200 \alpha |\Pi| \K$ and we take $\eps$ sufficiently small so that $\alpha \eps < 1/6$ , and the final equality comes from the definition of $\beta$.
\end{proof}

Now we proceed with establishing that the probability that $D$ is so large with respect to $\Err$ is relatively small. Let $\varphi_{u,v}(i)$ denote the value of $\varphi_{u,v}$ at the beginning of iteration $i$. Towards proving Lemma \ref{lem:hash-collision-bound-obliv}, define the random variable $X_{i,u,v}$ for all $(i,u,v)$ such that $\varphi_{u,v}(i) > 0$ as follows:
\[ X_{i,u,v} =
\begin{cases}
1 & \parbox[t]{0.6\columnwidth}{if hash collision occurs between $u$ and $v$ in iteration $i$\strut} \\
0 & \text{otherwise} \\
\end{cases}.
\]

Define the process $\psi_{u,v}$ as follows.
\begin{algorithm}[H]
\caption{The process $\psi_{u,v}$}
\begin{algorithmic}
\State $i \gets 1, \psi_{u,v}(1) \gets 0$
\ForAll{iterations $i$ from 1 to $100 |\Pi|$}
  \If {error occurs between $u$ and $v$ during iteration $i$, during any phase}
    \State $\psi_{u,v}(i+1) = \psi_{u,v}(i) + 6C_6$
  \ElsIf {$\psi_{u,v}(i) > 0$}
    \State $\psi_{u,v}(i+1) = \psi_{u,v}(i) + 6C_6 X_{i,u,v} - 5 (1 - X_{i,u,v})$
  \Else
    \State $\psi_{u,v}(i+1) = \psi_{u,v}(i)$
  \EndIf
\EndFor
\end{algorithmic}
\end{algorithm}

We remark that $\psi_{u,v}$ updates in such a way that it is always an upper bound on $\varphi_{u,v}$. We formalize this below.
\begin{lemma}
\label{lem:psi-lb}
For all iterations $i$ in \Alg and all $(u,v) \in E$, we have that $\psi_{u,v}(i) \geq \varphi_{u,v}(i)$, where $\varphi_{u,v}(i)$ denotes the value of the potential $\varphi_{u,v}$ at the beginning of iteration $i$.
\end{lemma}
\begin{proof}
We prove the claim by induction. Clearly it is true for iteration~$i=1$. Assume now that it is true for a certain~$i$. We will show that it is also true for iteration~$i+1$.

If $\varphi_{u,v}(i+1) = 0$, then the claim follows by taking $C_6$ to be divisible by 5 (since we can take $C_6$ to be sufficiently large and have no other constraints, this is doable). It is clear then that $\psi_{u,v}$ is always nonnegative by construction, and so $\psi_{u,v}(i+1) \geq \varphi_{u,v}(i+1)$.

Suppose there is an error between $u$ and $v$ in iteration $i$. We know that $\varphi_{u,v}$ increases by at most~$6C_6$ regardless of the number of errors or hash collisions in the entire iteration. This follows from Lemma \ref{lem:potential-rises-overall-1}, Proposition~\ref{prop:status-simulate}, and Claim \ref{claim:simple-Hpot-claim}: Lemma~\ref{lem:potential-rises-overall-1} and Proposition~\ref{prop:status-simulate} together show that $\varphi_{u,v}$ can increase by at most $5C_6$ in the Meeting Points phase, and Claim~\ref{claim:simple-Hpot-claim} shows that $\varphi_{u,v}$ can increase by at most $2C_3$ in all the other phases combined, so $C_6 \geq 2C_3$ yields the desired result. %
So we conclude that $\psi_{u,v}(i+1) \geq \varphi_{u,v}(i+1)$.

Now suppose that $\varphi_{u,v}(i+1) > 0$ and there is no error between $u$ and $v$ in iteration $i$. Then we must have that $\varphi_{u,v}(i) > 0$, since if $\varphi_{u,v}(i)=0$, we would have $B_{u,v}(i) = 0$ and, therefore there is no error or hash collision in iteration $i$ between $u$ and $v$ and hence the $\varphi_{u,v}$ cannot increase (Proposition~\ref{prop:status-simulate}, Lemma~\ref{lem:potential-rises-overall-1} for the Meeting Points phase, and Claim~\ref{claim:simple-Hpot-claim} for Flag Passing, Simulation, and Rewind phases). Furthermore, there can only be a hash collision at iteration $i$ between $u$ and $v$ if $\varphi_{u,v}(i) > 0$, which follows from Proposition \ref{prop:phi-nonnegative} and the observation that hash collisions can only happen if $B_{u,v}(i) > 0$ after iteration $i$. Hence, $\psi_{u,v}(i+1) = \psi_{u,v}(i) + 6C_6$ whenever there is a hash collision between $u$ and $v$, and $\psi_{u,v}(i+1) = \psi_{u,v}(i) - 5$ whenever there is no hash collision between $u$ and $v$. $\varphi_{u,v}$ increases by at most $5C_6$ in the presence of a hash collision and decreases by at least $5$ in the absence of one during the Meeting Points phase (Lemma \ref{lem:potential-rises-overall-1}). In the rest of the iteration, the potential can increase by at most $2C_3 < C_6$ if there was an error or hash collision on the link, and does not increase otherwise (Claim ~\ref{claim:simple-Hpot-claim}). Hence, $\varphi_{u,v}$ increases by at most $6C_6$ in the presence of an error or hash collision, and decreases by at least $5$ otherwise, and we conclude the proof.
\end{proof}

Let $\D$ denote the set of triples $(i,u,v)$ such that $\psi_{u,v}(i) > 0$. Then we claim that $D \leq |\D|$. This follows from the fact that $\psi_{u,v}(i) \geq B_{u,v}(i)$ for all $i$ (Lemma \ref{lem:psi-lb} and Proposition \ref{prop:phi-nonnegative}). Our goal will be to prove that $|\D|$ is not much larger than $\Err$, with high probability. Our strategy will be to use the fact that $\psi_{u,v}(i)$ is always nonnegative (again, an application of Lemma \ref{lem:psi-lb} and Proposition \ref{prop:phi-nonnegative}) and that the adversary cannot make too many errors, to argue that $\sum_{(i,u,v) \in \D} X_{i,u,v}$ must be bounded below by something relatively large. If there are many random variables $X_{i,u,v}$, then a Chernoff bound will let us argue that $\sum_{(i,u,v) \in \D} X_{i,u,v}$ should not be this large with high probability. However, the communication in the protocol is not a priori bounded\footnote{Except trivially, by the number of rounds of \Alg  times $m$, but this can be a factor $m$ more than $CC(\Pi)$.}: if there are many hash collisions or errors, then the parties might communicate more, which creates more budget for errors. So our first step is to bound the number of errors that the adversary can commit. 

\begin{lemma}
\label{lem:errs-bounded-obliv}
Suppose that the adversary is oblivious, and commits $\Err$ errors. Denote the hash collision probability in \Alg as~$p$, and suppose that $p < 1/30C_6$. Then,
\[
\Pr\left[ \Err \leq 200 \alpha \eps |\Pi| \ \bigvee \ \Err > \frac{\eps}{\K} \CC  \right ] \ge  1 - p^{\Omega(|\Pi|))},
\]
where the probability above is over the random seeds used for hashing throughout the protocol.
\end{lemma}
\begin{proof}[Proof of Lemma~\ref{lem:errs-bounded-obliv}]
Suppose $\Err > 200 \alpha \eps |\Pi|$ and $\Err \leq \frac{\eps}{\K} \CC$. This implies that $\CC > 200 \alpha |\Pi| \K$, and so applying Lemma \ref{lem:cc-bound} we see that $D > 1/(3\alpha\eps) \Err$. As we noted previously, $|\D| \geq D$ (consequence of Lemma~\ref{lem:psi-lb} and argument above), and so it suffices to bound the probability that $|\D|$ exceeds $1/(3\alpha\eps) \Err$. By taking $\eps$ small enough with respect to $\alpha$ and $C_6$, we can safely assume that $3\alpha\eps < 1/(10C_6)$, and therefore we can apply Lemma \ref{lem:bound-dangerous-obliv}, which tells us that the probability of this is at most $\exp(-\Omega(\Err \log(1/p) / \eps))$.
\end{proof}
\begin{lemma}
\label{lem:bound-dangerous-obliv}
Suppose the adversary is oblivious, and let $\Err$ be the number of errors committed. Let $p$ be the hash collision probability of $h$, and assume that $p <1/{30C_6}$. Let $k$ be some integer such that $k \geq 10C_6$. Let $\D$ denote the set of triples $(i,u,v)$ such that $\psi_{u,v}(i) > 0$. Then,
\[ \Pr[|\D| > k \Err ] < p^{\Omega\left(k \Err\right)}, \]
where the $\Omega$ hides constants on the order of $1/C_6$.
\end{lemma}
\begin{proof}
Consider running the protocol after fixing the adversary's errors with uniformly random seeds for the hashes. Assume that $|\D| > k \Err$, and let $\widetilde{\D} \subseteq \D$ be the subset of triples where no error occurs. Since the number of errors is $\Err$, we get that $|\widetilde{\D}| \geq (k - 1) \Err$. 

We will argue that the fraction of hash collisions required in the $|\widetilde{\D}|$ (possibly) dangerous triples is far too large.

Define 
\[\psi \eqdef \sum_{(u,v) \in E} \psi_{u,v}\]
Note that $\psi$ is always nonnegative by design, since each of the $\psi_{u,v}$'s are nonnegative. Consider $\psi(100 |\Pi|+1)$, that is, the value of $\psi$ immediately after the final $100 |\Pi|$-th iteration of \Alg. Then, recalling the definition of $\psi_{u,v}$, we can upper bound $\psi(100 |\Pi|+1)$ as follows.
\begin{align} 
\nonumber
0 &\leq \psi(100 |\Pi| + 1) \\
&\leq 6C_6 \cdot \Err + \sum_{(i,u,v) \in \wt{\D}} (6C_6 X_{i,u,v} - 5(1-X_{i,u,v})).
\label{eq:bound-dangerous-obliv-1}
\end{align}
Hence,
\begin{equation}
\sum_{(i,u,v) \in \wtD} 6C_6 X_{i,u,v} - 5(1-X_{i,u,v}) \geq -6C_6 \cdot \Err.
\end{equation}

Now we claim that $\sum X_{i,u,v} \geq \frac{4}{6C_6+5}|\wtD|$.
\begin{align}
\sum_{(i,u,v) \in \wtD} \left( 6C_6 X_{i,u,v} -5 (1 - X_{i,u,v}) \right) &\geq -6C_6 \cdot \Err \nonumber \\
(6C_6 + 5)\left( \sum_{(i,u,v) \in \wtD} X_{i,u,v} \right) - 5|\wtD| &\geq -6C_6 \cdot \Err \nonumber \\
 \sum_{(i,u,v) \in \wtD} X_{i,u,v} &\geq \frac{4}{6C_6 + 5}|\wtD|, \label{eq:bound-danger-1}
\end{align}
where in the last line we use the fact that $|\wtD| \geq (k-1) \cdot \Err \geq 6C_6 \cdot \Err$. 

To summarize, we can now conclude that $|\wtD| \geq (k-1) \cdot \Err$ only if:
\begin{enumerate}
\item There are $|\wtD| \geq (k-1) \cdot \Err$ random variables $X_{i,u,v}$ such that $\psi_i(u,v) > 0$ and no error occurs between $u$ and $v$ in this iteration.
\item At least $4 / (6C_6 + 5)$ fraction of these random variables have value $1$, indicating a hash collision.
\end{enumerate}

To proceed with showing that the probability of this event is at most $p^{\Omega(k \Err)}$, we tie the probability on the left-hand side of Eq.~\eqref{eq:bound-danger-1} to the outcome of a certain coin-flipping experiment. Specifically, we will use Observation~\ref{obs:oblivious} to show that the probability of the above two events is upper bounded by the probability that we sample $|\wtD|$ i.i.d $\text{Ber}(p)$ random variables and at least $4/(6C_6 + 5)$ fraction of those variables end up as 1 (note that $|\wtD|$ is still a random variable here, but is defined to be the same in the process we care about and the coin-flipping process). This can be seen with a coupling argument. 

Order the triples $(i,u,v)$ in increasing order by iteration (arbitrarily by $(u,v)$). We define a Markov chain that counts hash collisions in $\wtD$, where the states are defined by the randomness seen so far in the protocol, the number of hash collisions that randomness implies are in $\wtD$, and whether the next triple to be seen is in $\wtD$. The transitions correspond to observing the randomness in the next triple we haven't seen. This Markov chain is well-defined: since the adversary is oblivious, whether or not a triple is in $\wtD$, as well as whether or not that triple has a hash collision, is fixed by the randomness up to that point\footnote{The observant reader will note that, technically, all we need in the states is the randomness seen so far in the protocol to make this chain well-defined, but the extra information may be useful to keep in mind. The presence of a hash collision depends on the transcripts, which depends on the previous randomness, which is why this chain needs to keep track of all the randomness so far - if this were unnecessary, we could make our state space just keep track of the hash collisions seen so far in $\wtD$.}.

We couple the Markov chain above to the process of generating $|\wtD|$ $\text{Ber}(p)$ random variables in the following way: whenever we reach a triple in $\wtD$ and observe the randomness used in it, we generate one $\text{Ber}(p)$ random variable (independently of the previous generated bits) and add it to a running sum. We will couple these processes such that the triple in $\wtD$ has a hash collision only if the $\text{Ber}(p)$ random variable turns out 1. This proves the desired result, that the number of 1's among the $|\wtD|$ $\text{Ber}(p)$ random variables stochastically dominates the number of hash collisions in this process.\footnote{For more background on coupling, see Chapter 5 of \cite{levin2017markov}.}

Observation~\ref{obs:oblivious} is the key part of this argument. Specifically, it shows that the hash collision probability for any transcripts $T_{u,v}(i) \neq T_{v,u}(i)$ is exactly $p$, and that additionally conditioning on previous hash collisions does not change this fact; the distribution of the random seed $S_{i,u,v}$ conditioned on $T_{u,v}(i)$, $T_{v,u}(i)$, and $\{ X_{j,a,b} : j \leq i-1, (a,b) \in E\}$ is still uniform. This is because all of these random variables are in fact a function of the seeds in previous iterations, and $S_{i,u,v}$ is independently sampled from these.

Indeed, the property we want is satisfied at the beginning (when there are 0 hash collisions in $\wtD$ so far, and 0 coins have been flipped). Consider taking a step in the Markov chain we defined. There are two cases: either the current state of the chain corresponds to a triple in $\wtD$, or it does not. If it does not correspond to a triple in $\wtD$, then we do not flip a coin, and the property remains satisfied. If it does correspond to a triple $(i,u,v)$ in $\wtD$, Observation~\ref{obs:oblivious} tells us that the probability of a hash collision is at most $p$: it is exactly $p$ if the transcripts are different, and 0 otherwise. Hence, we can couple this event to the event of observing a $\text{Ber}(p)$ random variable, such that the hash collision happens only if the $\text{Ber}(p)$ random variable turns up 1. So, the total number of 1's among the $\text{Ber}(p)$ random variables still stochastically dominates the number of hash collisions seen within $\wtD$, satisfying the desired property.

To summarize, to prove the result, it suffices to show that the probability that $4 / (6C_6 +5)$ fraction of at least $(k-1) \cdot \Err$ i.i.d $\text{Ber}(p)$ random variables are 1 is sufficiently small. Let $Y$ denote the event that $|\wtD| \geq (k-1) \cdot \Err$ and at least $4 / (6C_6 +5)$ fraction of the $\text{Ber}(p)$ random variables come up~1. Let $Y_d$ denote the event that at least $4 / (6C_6 +5)$ fraction of the $|\wtD|$ $\text{Ber}(p)$ random variables come up 1 conditioned on $|\wtD| = d$. Then we have that 
\begin{align*}
\Pr[Y] = \sum_{d = (k-1)\Err}^{\infty} \Pr[Y_d]\Pr[|\wtD| = d].
\end{align*}
Note that $Y_d$ is exactly the event that we generate $d \geq (k-1) \Err$ i.i.d $\text{Ber}(p)$ random variables and at least $4 / (6C_6 +5)$ fraction of them turn up 1. We now upper bound this event by~$p^{\Omega(d)}$, proving the desired result.

Now that we are dealing with a fixed number of i.i.d $\text{Ber}(p)$ random variables, we can do this using a Chernoff bound. Let $A_1, \ldots, A_d$ be iid $\text{Ber}(p)$ random variables. By taking $p < 1/(6C_6 + 5)$, we note that 
\[
\Pr\left[\sum_{i=1}^{d} A_{i} \geq \frac{4}{6C_6 + 5}d\right] 
\leq \Pr \left[\frac{1}{d}\sum_{i=1}^{d} A_{i} \geq p + \frac{3}{6C_6 + 5}\right].
\]
We bound this probability by $p^{\Omega(d)}$ in Claim \ref{claim:additive-chernoff}, by applying it with $\gamma = 3/(6C_6 + 5)$.
\end{proof}
\begin{claim}
\label{claim:additive-chernoff}
Fix a constant $0 < \gamma < 1/2$. Fix $d > 0$, and let $A_1, \ldots, A_d$ be i.i.d $\textrm{Ber}(p)$ random variables such that $p < 1 - 2 \gamma$. Then
\[ 
\Pr \left[\frac{1}{d}\sum_{j=1}^{d} A_j \geq p + \gamma \right] 
\leq p^{\Omega(d)} 
\]
where the $\Omega$ hides constants on the order of $\gamma$.
\end{claim}
\begin{proof}%
From a Chernoff bound, we get that
\begin{equation*}
\Pr\left [\frac{1}{d}\sum_{j=1}^{d} A_j \geq p + \gamma\right] \leq \exp\left(-D\left(p + \gamma \ \middle|\middle|\  p\right) \cdot d\right)
\end{equation*}
where $D(x||y) = x \ln\left( \frac{x}{y}\right) + (1-x) \ln\left( \frac{1-x}{1-y}\right)$ is the Kullback-Leibler divergence between two Bernoulli random variables. We claim that $D(p + \gamma || p) \geq \Omega(\log(1/p))$.
\begin{align*}
&D(p + \gamma || p) \\
&= (p + \gamma)\ln\left( \frac{p + \gamma}{p}\right) + (1-p-\gamma) \ln\left( \frac{1-p-\gamma}{1-p}\right) \\
&\geq \gamma \ln\left( \frac{\gamma}{p}\right) + (1-p-\gamma)\left(\frac{-2\gamma}{1-p}\right) \\
&\geq \gamma \ln\left( \frac{\gamma}{p}\right) - 2\gamma \\
&\geq \gamma \left(\ln\left( \frac{\gamma}{p}\right) - 2\right)
\end{align*}
where in the second line we use the inequality that $-x \leq \ln(1 - x/2)$ for $x \in (0,1)$, using the fact that $p < 1 - 2\gamma$.
\end{proof}

From the above lemmas, the proof of Lemma \ref{lem:hash-collision-bound-obliv} follows easily.
\begin{proof}[Proof of Lemma \ref{lem:hash-collision-bound-obliv}]
By Lemma \ref{lem:errs-bounded-obliv}, the probability of having more than $200 \alpha \eps |\Pi|$ errors in the protocol while also having $\Err \leq \frac{\eps}{\K} \CC$ is at most $p^{\Omega(|\Pi| / \eps)}$. The probability that the number of errors is smaller than $200 \alpha \eps |\Pi|$ and simultaneously the number of dangerous triples exceeds $k \cdot (200 \alpha \eps) |\Pi|$ is at most $p^{\Omega(k \eps |\Pi|)}$
, by Lemma~\ref{lem:bound-dangerous-obliv}.\footnote{Note that the $k$ in the statement of Lemma~\ref{lem:hash-collision-bound-obliv} and the $k$ with which we invoke Lemma~\ref{lem:bound-dangerous-obliv} here are different - technically, we invoke Lemma~\ref{lem:bound-dangerous-obliv} with its $k$ set to $k(200 \alpha \eps |\Pi|) / \Err$. We settle for this abuse of notation because $k$ plays the same role conceptually in Lemmas~\ref{lem:hash-collision-bound-obliv} and~\ref{lem:bound-dangerous-obliv}.} 
Hence, the union of these two events has probability at most 
\[
p^{\Omega(k \eps |\Pi|)},
\]
which is the only place where we use that $k = O(1/\eps)$. %
\end{proof}

Now that we have bounded the communication and number of hash collisions, we are ready to prove Theorem~\ref{thm:oblivious-adv-crs}.

\changelocaltocdepth{3}
\subsection{Completing the proof of Theorem~\ref{thm:oblivious-adv-crs}}\label{sec:proof-main-thm}
Recall that we set $K := m$ in InitializeState (Algorithm~\ref{alg:initialize-state}). Throughout this analysis, we assume that the number of errors that the adversary commits is $\Err$ such that $\Err \leq (\eps / m)\CC$, as this is what is implied by the adversary having rate $\eps / m$. Note that our final goal is to show that the probability that both $\Err \leq (\eps / m) \CC$ and that the protocol is not simulated correctly is small; at most $\exp(-\Omega(|\Pi|))$.

Lemma \ref{lem:phi-rises-enough} shows that in every iteration of \Alg, $\phi$ \emph{always} increases by at least $\K = m$. Hence, after at the end of the simulation after $100 |\Pi|$ iterations, we know that $\phi \ge 100 |\Pi| m$.

We conclude the argument by establishing that, when the protocol ends and $\phi \geq 100|\Pi|m$, the parties are done simulating the initial protocol correctly with high probability. We do this by appealing to the following claim, written for general $\K$, and apply it with $\K=m$ to finish the proof.

\begin{claim}
\label{claim:final-claim}
Suppose that $\phi \geq 100 |\Pi| \K$ at the end of \Alg. Furthermore suppose that $\CC$ and $\EHC$ satisfy the positive outcome of Lemma \ref{lem:bounding-collisions-main}. Specifically, that $\CC \leq 200\alpha |\Pi| \K$ and that $\EHC \leq (k+1) (200 \alpha \eps) |\Pi|$, with $k := (20C_6\alpha\eps^* / \eps) - 1$ for some $0 < \eps^* \leq 1/(4000 C_6 C_7 \alpha)$, such that $\eps$ is sufficiently smaller than $\eps^*$. 

Conditioned on these events, the parties have simulated the underlying protocol $\Pi$ correctly.
\end{claim}
\begin{proof}
First note that it is valid to invoke Lemma \ref{lem:bounding-collisions-main} with $k := (20C_6\alpha\eps^* / \eps) - 1$, since taking $\eps$ sufficiently small wrt $\eps^*$ makes $k$ be large enough for the lemma to apply. 

Recall our definition of the potential $\phi$:
\[
\phi = \sum_{(u,v) \in E} \left( (\K / m) G_{u,v} - \K \cdot \varphi_{u,v} \right) - C_1 \K B^* + C_7 \K \cdot \EHC.
\]

Since $\eps^*$ is sufficiently small, we get that 
\begin{align*}
C_7 \K \cdot \EHC \le C_7\K \cdot 4000 C_6  \alpha \eps^* |\Pi|  \le |\Pi| \K.
\end{align*}

Furthermore, the term $\sum  - \K \cdot \varphi_{u,v}$ is nonpositive, due to the fact that $\varphi_{u,v}$ is non-negative (Proposition \ref{prop:phi-nonnegative}). Therefore, we get that
\begin{equation}
\label{eq:put-all-together-1}
\sum_{(u,v) \in E} \frac{\K}{m}G_{u,v}  - C_1 \K B^* \geq 99 |\Pi| \K.
\end{equation}

Recall that $B^* = H^* - G^* \geq \max_{(u,v) \in E} (G_{u,v}) - \min_{(u,v) \in E}(G_{u,v})$. By plugging this into Eq.~\eqref{eq:put-all-together-1}, and recalling that $C_1\ge 2$, 
we get
\begin{align*}
 &99 |\Pi| \K 
 \leq \sum_{(u,v) \in E} \frac{\K}{m}G_{u,v}  - C_1 \K B^* \\
&\leq \K  \max_{(u,v) \in E} (G_{u,v}) - C_1 \K \left(\max_{(u,v) \in E} (G_{u,v}) - \min_{(u,v) \in E}(G_{u,v})\right) \\
&\leq \K  \max_{(u,v) \in E} (G_{u,v}) - \K\left (\max_{(u,v) \in E} (G_{u,v}) - \min_{(u,v) \in E}(G_{u,v}) \right) \\
&= \K \min_{(u,v) \in E}(G_{u,v})
\end{align*}
Hence, we conclude that $\min_{(u,v) \in E}(G_{u,v}) \geq 99 |\Pi|$. 
Therefore, each pair of parties have simulated $\Pi$ correctly for at least $|\Pi|$ chunks, which suffices to compute $\Pi$ correctly.

\end{proof}

To finish, we recall that Lemma \ref{lem:bounding-collisions-main} shows that with probability $1 - p^{\Omega(k \eps |\Pi|)}$, one of the following happens:
1) $\CC \leq 200\alpha |\Pi| m$ and $\EHC \leq (k+1) (200 \alpha \eps) |\Pi|$, or
2) $\Err > (\eps / m) \CC$. 

Since we take $p$ and $\alpha$ to be constants, the first case tells us that $\CC = O(|\Pi|m) = O(\CC(\Pi))$ with probability $1 - \exp(- \Omega(|\Pi|))$. In the first case, Claim \ref{claim:final-claim} also tells us that the parties have simulated $\Pi$ correctly. 

Finally, we recall that in Claim \ref{claim:final-claim}, we apply Lemma \ref{lem:bounding-collisions-main} with $k = \Theta(\eps^* / \eps)$. So we get that $1 - p^{\Omega(k \eps |\Pi|)} = 1 - \exp(- \Omega(\eps^* |\Pi|))$. Since $\eps^*$ is a fixed constant in terms of $\alpha, C_6$, and $C_7$, we can absorb it into the $\Omega$ to conclude that, with probability $1 - \exp(-\Omega(|\Pi|))$, the parties either simulate correctly or the event that $\Err > (\eps / m)\CC$ occurs. 
This finishes the proof of Theorem~\ref{thm:oblivious-adv-crs}.

\section{Acknowledgments}
We would like to thank the anonymous referees, who pointed out numerous inconsistencies and unclear phrasings. We think that addressing these has greatly improved the understandibility of this paper.

\bibliographystyle{IEEEtranS}
\bibliography{network}

\appendix
\section{The Meeting Points Mechanism}
\label{app:meetingpoints}
In this appendix we define and analyze the meeting points mechanism~\cite{haeupler14} and its respective potential~$\varphi_{u,v}$. Certain parts of the analysis mostly repeat~\cite{haeupler14} and are included here for completeness. The main difference from~\cite{haeupler14} is that the meeting points mechanism is interleaved over several iteration of the simulation protocol. This brings two potential difficulties. First, $\varphi_{u,v}$ may change outside the meeting-point phase. Second, the transcript may change \emph{while the meeting-points mechanism is still in progress}. We eliminate the latter by fixing a party's transcript until the meeting points mechanism reports the transcripts of both parties are consistent. The former requires a more careful analysis of~$\varphi_{u,v}$, and is handled in Claim~\ref{claim:simple-Hpot-claim}. 

\label{sec:meeting-points}
\subsection{Meeting Points Protocol Between Two Parties}
Below we describe the meeting points protocol that parties do pairwise. We write the algorithm as it is performed by some party $u$ with one of its neighbors $v$. In all variables below, we will drop the subscript of $(u,v)$ but it is implied. Specifically, $k$ below denotes $k_{u,v}$, and the same is true for $E$, $T$, $mpc1$, $mpc2$, and $status$.
\begin{algorithm}
\caption{\textbf{MeetingPoints}($u$,$v$, $S_{i,u,v}$, $\K$)}
\label{alg:meeting-points}
\begin{algorithmic}[1]
\small
 \Statex // Initialization:
  \State Method called by $u$, with $v \in N(u)$. $S_{i,u,v} \in \{ 0,1\}^{\Theta(|\Pi|\K)}$ is a large random seed to be split up and used for hashing. $\Pi \gets$ protocol to be simulated.
  \State $h \gets$ inner product hash family (Definition~\ref{def:ip-hash}) w/ input length $\Theta(|\Pi|\K)$, $p= 2^{-\Theta(\K / m)}$ sufficiently small, $o = \Theta(\K / m)$, $s = |S_{i,u,v}|/10$.
  \State $(S_1, S_2, \ldots, S_{10}) \gets S_{i,u,v}$. $S_{i,u,v}$ is split into ten seeds. Wlog we assume that $S_1, \ldots, S_5$ are for the hashes it sends, and $S_6, \ldots, S_{10}$ are for the hashes it uses for comparisons (so $v$ is using $S_6, \ldots S_{10}$ for the hashes it sends, and $S_1, \ldots, S_5$ for comparisons, respectively).
  \State $k,E, mpc1, mpc2 \gets 0$
  \Statex
  \Statex // Execute per activation:
  \State $k \gets k+1$
  \State $\widetilde{k} \gets 2^{\lfloor \log k \rfloor}$. Let $c$ be the largest integer such that $c\widetilde{k} \leq |T_{u,v}|$.
  \State $T_1 \gets T_{u,v}[1:c\widetilde{k}], T_2 \gets T_{u,v}[1:(c-1)\widetilde{k}]$
  \State Send $(h_{S_1}(k), h_{S_2}(T_1), h_{S_3}(T_1), h_{S_4}(T_2), h_{S_5}(T_2))$ to neighbor $v$.
  \State $(H_k, H_{T_1}^{(1)}, H_{T_2}^{(1)}, H_{T_1}^{(2)}, H_{T_2}^{(2)})$\strut
  \Statex $\qquad\gets (h_{S_6}(k), h_{S_7}(T_1), h_{S_8}(T_2), h_{S_9}(T_1), h_{S_{10}}(T_2))$\strut
  \State Receive $(H_k', H_{T_1}^{(1)'}, H_{T_1}^{(2)'}, H_{T_2}^{(1)'}, H_{T_2}^{(2)'})$ from our neighbor $v$. \label{step:recHash}
  \If {$H_k \neq H_k'$}
    \State {$E \gets E+1$}
  \EndIf
  \If {$k=1, E=0$, and $H_{T_1}^{(1)} = H_{T_1}^{(1)'}$} \label{step:HashesMatch}
    \State {$k \gets 0$}
    \State {$status \gets$ ``simulate''}
    \Return $status$
  \EndIf
  \If {$H_{T_1}^{(1)} = H_{T_1}^{(1)'}$ or $H_{T_1}^{(2)} = H_{T_2}^{(1)'}$}
    \State {$mpc1 \gets mpc1 + 1$}
  \ElsIf {$H_{T_2}^{(1)} = H_{T_1}^{(2)'}$ or $H_{T_2}^{(2)} = H_{T_2}^{(2)'}$}
    \State {$mpc2 \gets mpc2 + 1$}
  \EndIf
  \If {$2E \geq k$}~\label{step:resetMP}
    \State {$k \gets 0, E \gets 0, mpc1 \gets 0, mpc2 \gets 0$}
    \State {$status \gets$ ``meeting points''}
  \ElsIf {$k = \widetilde{k}$}
    \If {$mpc1 > 0.4k$}
      \State {$T \gets T_1$}
      \State {$k \gets 0, E \gets 0$}
    \ElsIf {$mpc2 > 0.4k$}
      \State {$T \gets T_2$}
      \State {$k \gets 0, E \gets 0$}
    \EndIf
    \State {$mpc1 \gets 0, mpc2 \gets 0$}
    \State {$status \gets$ ``meeting points''}
  \Else
    \State {$status \gets$ ``meeting points''}
  \EndIf \\
  \Return $status$
\end{algorithmic}
\end{algorithm}

 Roughly speaking, there are up to two types of actions performed in each round of meeting points. The pair of parties first send each other hashes of their truncated transcripts, and increment their respective $E$, $T_1$, or $T_2$ counters if applicable. In keeping with Haeupler's paper we will call this the \emph{verification} phase of meeting points. Note that the verification phase \emph{always} occurs during a round of meeting point exchange.

After exchanging hashes, the parties judge whether or not to take further action based on the values of $E$, $T_1$, and $T_2$. We will call this the \emph{transition} phase of meeting points. For example, if $2E_{u,v} > k_{u,v}$, then party $u$ will set all its variables in the meeting points computation to 0. We will call these transitions \emph{reset} transitions\footnote{This is called an error transition in~\cite{haeupler14}.}. Otherwise, if $mpc1_{u,v} > 0.4k_{u,v}$, party $u$ will transition to meeting point 1, and similarly for meeting point 2. These will be called \emph{meeting point} transitions.

We have the parties use separate seeds for each hash comparison, and so they take the large shared seed between them and split it into many smaller seeds. The reason for this is somewhat technical; it makes the events of hash collisions for the different comparisons independent (when the seeds themselves are independent), which will be useful in our analysis of removing the common random string in the second part of this work~\cite{GKR-2}\footnote{While it helps with our analysis, is not clear that using separate seeds is \emph{necessary} to remove the common random string.}.
It suffices to use a single seed in Meeting Points for the usage in Section 4 (with a common random string). 

\subsection{Notation}
We establish some notation that will be used in this section. \gnote{removed a bullet pt for reviewer 2}
\begin{itemize}
\item $E_{u,v}, k_{u,v}$, $mpc1_{u,v}$, and  $mpc2_{u,v}$ are all defined as the value of the corresponding variable that party $u$ has for the communication link $(u,v)$. We can also define these with $v$ coming first in the subscript (e.g. $E_{v,u}$); these will correspond to the value that party $v$ has for the link $(u,v)$.
\item $WM_{u,v}$ corresponds to the number of wrong matches or mismatches that contribute to the current values of $mpc1_{u,v}$ and $mpc2_{u,v}$ (the counters for party $u$) on link $(u,v)$. That is, if $(T_1)_{u,v} \in \{ (T_1)_{v,u}, (T_2)_{v,u}\}$ but party $u$ does not increment $mpc1$ due to an error, then we increment $WM_{u,v}$ by 1. Similarly, if  $(T_1)_{u,v} \not\in \{ (T_1)_{v,u}, (T_2)_{v,u}\}$ but party $u$ increments $mpc1$ due to an error or hash collision, we increment $WM_{u,v}$ by 1. Define $WM_{v,u}$ similarly, but for the increments and non-increments of party $v$.
\item Let $k^{\{u,v\}} = k_{u,v} + k_{v,u}$. Define $E^{\{u,v\}}, mpc1^{\{u,v\}}, mpc2^{\{u,v\}}$, and $WM^{\{u,v\}}$ similarly.
\item Given some number $V$ that depends on the transcript $T$, define $\Delta(V)$ to be the change in $V$ that results after one invocation to Meeting Points (Algorithm \ref{alg:meeting-points}) for all pairs of adjacent parties (i.e. in the Meeting Points phase of Algorithm~\ref{alg:robust-protocol}). Throughout the analysis, we assume we have fixed some particular invocation to the Meeting Points algorithm (i.e. corresponding to a fixed iteration $i$), and we want to understand how this invocation affects various statistics about the transcript $T$.
\item Similarly, given a number $V$ that depends on the transcript $T$, define $\Delta_{u,v}(V)$ to be the change in $V$ that results in parties $u$ and $v$ running Meeting Points (Algorithm \ref{alg:meeting-points}) with each other, and no other pair of parties making changes.
\item When it is understood that we are only talking about the interaction between a pair of parties $u$ and $v$, we will drop the superscript $\{u,v\}$ off terms such as $k^{\{u,v\}}, E^{\{u,v\}}$, $WM^{\{u,v\}}$, etc. with the understanding that we are only talking about this pairwise interaction.
\end{itemize}
\subsection{Potential Analysis}
\label{sec:app-mp-potential}

Following the paper of Haeupler~\cite{haeupler14}, we define $\varphi_{u,v}$ as follows. Let $1 < C_1 < C_2 < C_3 < C_4 < C_5 < C_6 < C_7$, where each $C_i$ is selected to be sufficiently larger than $C_{i-1}$ (or 1 if $i=1$).
\begin{equation}\label{eqn:Hpot}
 \varphi_{u,v} = \begin{cases}
\begin{aligned} &C_3 \cdot B_{u,v} - C_2 \cdot k^{\{u,v\}} \\ &+ C_5 \cdot E^{\{u,v\}} + 2C_6  \cdot WM^{\{u,v\}}\end{aligned} &\text{ if }k_{u,v} = k_{v,u}\Tstrut \\[2.5ex]
\Tstrut\begin{aligned} &C_3 \cdot B_{u,v} + 0.9C_4 \cdot k^{\{u,v\}} \\  &- C_4 \cdot E^{\{u,v\}} + C_6 \cdot WM^{\{u,v\}}\end{aligned} &\text{ if }k_{u,v} \neq k_{v,u}
\end{cases}
\end{equation}

Recall that we define our final potential function as follows:
\[
\phi = \sum_{(u,v) \in E} (\K/m)G_{u,v} - \K \cdot \varphi_{u,v} - C_1 \K B^* + C_7 \K \cdot EHC.
\]

We start with some simple claims about $\varphi_{u,v}$ that we use directly in the main proof. The proof of following claim (Claim \ref{claim:simple-Hpot-claim}), unlike the other claims in this section, requires knowledge of the robust protocol (Algorithm \ref{alg:robust-protocol}) beyond the definition of $\varphi_{u,v}$ and the Meeting Points protocol (Algorithm \ref{alg:meeting-points}).
\begin{claim}
\label{claim:simple-Hpot-claim}
The potential $\varphi_{u,v}$ does not change in the Flag Passing phase. In each of the Rewind and Simulation phases, it changes by at most $C_3$. In the absence of errors or hash collisions between $u$ and $v$ in the \emph{iteration as a whole}, $\varphi_{u,v}$ does not increase in the Rewind or Simulation phases.
\end{claim}
\begin{proof}
The only term in $\varphi_{u,v}$ that changes in these phases is $B_{u,v}$. $B_{u,v}$ does not change in the Flag Passing phase, can change by at most 1 in the Simulation phase, and can change by at most 1 in the Rewind phase. This establishes the first part of the claim.

Now we establish the second part of the claim. First we consider the Simulation phase. At a high level, $B_{u,v}$ can increase in the Simulation phase if one of the following is true: 1) $u$ and $v$ are in disagreement, but both decide to simulate anyway, 2) The adversary puts an error between $u$ and $v$ in the Simulation phase, or 3) $u$ decides not to simulate $\Pi_{u,v}$, but $v$ decides to simulate $\Pi_{v,u}$. We now formalize the three cases and prove that all of them require an error or hash collision between $u$ and $v$ somewhere in the iteration.

In the first case, $\Pi_{u,v} \neq \Pi_{v,u}$ but $netCorrect_u = netCorrect_v = 1$. Note that this implies that $status_{u,v} = status_{v,u} = \text{``simulate''}$, since otherwise one of the parties would have $netCorrect = 0$. This means there was an error or hash collision between $u$ and $v$ in the Meeting Points phase to make them think that their transcripts matched. In the second case, there was an error between $u$ and $v$ in the Simulation phase.
In the third case, note that $v$ simulates in $\Pi_{v,u}$, so we know that $netCorrect_v = 1$ and $v$ never received $\perp$ from $u$ in the Simulation phase. Furthermore, $u$ does not simulate. This can happen due to one of two reasons. The first case is that $netCorrect_u = 0$, and so $u$ did not want to simulate after the Flag Passing phase. In this case, $u$ would have sent $\perp$ to $v$. Since $v$ did not receive it, the adversary must have deleted it. The second case is that $u$ had $netCorrect_u = 1$, but received a $\perp$ on the link $(u,v)$. However, $netCorrect_v = 1$, so $v$ would not have sent this $\perp$. Hence, it was inserted by the adversary. Either way, there is an error on the link $(u,v)$ during the Simulation phase.

Now we consider the Rewind phase. If neither $u$ nor $v$ send or receive a rewind message on $(u,v)$, then it is clear that $B_{u,v}$ is unchanged. We assume that any rewind sent on the link $(u,v)$ reaches the recipient. If this is not the case, then there was a deletion on the link $(u,v)$ and we are done. Similarly, we assume that any rewind received on the link $(u,v)$ was sent by the other party (otherwise it was an insertion). We break the remainder of the proof into cases depending on which parties want to send rewinds on the link $(u,v)$.

Suppose that $status_{u,v} = status_{v,u} = \text{``simulate''}$, $alreadyRewound_{u,v} = alreadyRewound_{v,u} = 0$. Then the following actions all happen together in the same round or not at all: $u$ (resp. $v$) sends ``rewind'' to $v$ (resp. $u$), $u$ truncates $\Pi_{u,v}$, $v$ truncates $\Pi_{v,u}$, and $alreadyRewound_{u,v}$ and $alreadyRewound_{v,u}$ are set to 1. If all these actions happen, then $B_{u,v}$ does not increase. This is because $\max\{ |\Pi_{u,v}|, |\Pi_{v,u}|\}$ falls by one, and $G_{u,v}$ falls by at most 1. Furthermore, after these actions happen, $alreadyRewound_{u,v} = alreadyRewound_{v,u} = 1$, so $u$ and $v$ will not rewind on $(u,v)$ anymore.

If $status_{u,v} = \text{``meeting points''}$ or $alreadyRewound_{u,v} = 1$, then $u$ will not send a rewind message to $v$ nor take any action if it receives a rewind from $v$. Therefore, if $v$ does not send a rewind message to $u$, then $B_{u,v}$ is unchanged in this Rewind phase. If $v$ does send a rewind message to $u$, then we know that just before the message was sent we had $status_{v,u} = \text{``simulate''}$ and $alreadyRewound_{v,u} = 0$. If $status_{u,v} = \text{``meeting points''}$, then there must have been a hash collision or error between $u$ and $v$ in the Meeting Points phase, otherwise they would both have the same status. Otherwise if $alreadyRewound_{u,v}=1$, then $u$ has already truncated $\Pi_{u,v}$ by one chunk. So the net change to $B_{u,v}$ from the truncation of $\Pi_{u,v}$ and $\Pi_{v,u}$ that have happened in this rewind phase is nonpositive, as established in the previous paragraph. Furthermore, after this rewind is sent, both parties have $alreadyRewound=1$, and will not do anything further on this link in the rewind phase.
\end{proof}

\begin{proposition}
\label{prop:phi-nonnegative}
At the beginning of any iteration $i$ of the robust protocol (Algorithm \ref{alg:robust-protocol}), the following is true for all pairs $(u,v) \in E$:
\[ 0 \leq B_{u,v} \leq \varphi_{u,v}.
\]
As a corollary, $\sum_{(u,v) \in E} \varphi_{u,v} \geq \sum_{(u,v) \in E} B_{u,v}$.
\end{proposition}
\begin{proof}
If $k_{u,v} = k_{v,u}$, follows from fact that we can take $C_3 > 8C_2 + 1 $. Thus, if $k_{u,v}$ is large enough such that $C_2 k^{\{u,v\}} > (C_3-1) B_{u,v}$, this means that we have $k_{u,v} = k_{v,u} > 4 B_{u,v}$. But we know that when both $k_{u,v}$ and $k_{v,u}$ are larger than $2B_{u,v}$, both parties are including hashes of their pairwise transcript truncated below the chunk $G_{u,v}$, which should match. So this means the fact that $k_{u,v}$ and $k_{v,u}$ increased so much more without the parties doing a meeting point transition means that there were many mismatches due to errors - specifically, $WM^{\{u,v\}} > 0.6(1/2)k^{\{u,v\}}$. Hence, in this case, the sum of all the terms that do not include $B_{u,v}$ is nonnegative.

If $k_{u,v} \neq k_{v,u}$, then this follows from the fact that $k^{\{u,v\}} \geq 2E^{\{u,v\}}$ so $0.9C_4 k^{\{u,v\}} \geq C_4 E^{\{u,v\}}$.
\end{proof}

Ideally, we would like to lower bound $\Delta(\phi)$, the amount that $\phi$ changes overall in the Meeting Points phase, by $\sum_{u,v} \Delta_{u,v}(\phi)$, the sum of the amounts that $\phi$ changes just due to the Meeting Points interaction between each pair of adjacent parties $(u,v)$. The reason this is nontrivial is that $\phi$ includes the global term $B^*$, which is a max over terms from each party rather than a sum. 
To this end, we define a lower bound on $\Delta_{u,v}(\phi)$ as follows:
\begin{align*}
&\widetilde{\Delta_{u,v}(\phi)} \\
&= (\K/m)\Delta_{u,v}(G_{u,v}) - \K \cdot \Delta_{u,v}(\varphi_{u,v}) + C_1 \K \cdot \Delta_{u,v}(G^*).
\end{align*}

\begin{claim}
\label{claim:total-potential-rises}
 $\Delta(\phi) \geq \sum_{(u,v) \in E} \left( \widetilde{\Delta_{u,v}(\phi)} + C_7 \K \Delta_{u,v}(EHC) \right)$.
\end{claim}
\begin{proof}
\begin{align*}
\Delta(\phi) 
&= (\K / m) \Delta(\sum G_{u,v}) - \K \cdot \Delta(\sum \varphi_{u,v}) \\ & \quad- C_1 \K \cdot \Delta(B^*) + C_7 \K \cdot \sum \Delta_{u,v}(EHC)\\
&\geq \sum (\K / m) \Delta(G_{u,v}) - \K \cdot \sum \Delta(\varphi_{u,v}) \\ & \quad + C_1 \K \cdot \Delta(G^*)  + C_7 \K \cdot \sum \Delta_{u,v}(EHC)\\
&= \sum (\K / m) \Delta(G_{u,v}) - \K \cdot \sum \Delta(\varphi_{u,v}) \\ & \quad + C_1 \K \cdot \min_{(u,v) \in E}(\Delta_{u,v}(G^*)) + C_7 \K \cdot \sum \Delta_{u,v}(EHC)\\
&\geq \sum (\K / m) \Delta(G_{u,v}) - \K \cdot \sum \Delta(\varphi_{u,v}) \\ & \quad + C_1 \K \cdot \sum \Delta_{u,v}(G^*) + C_7 \K \cdot \sum \Delta_{u,v}(EHC) \\
&= \sum_{(u,v) \in E} \left( \widetilde{\Delta_{u,v}(\phi)} + C_7 \K \Delta_{u,v}(EHC) \right).
\end{align*}
The second line follows from the fact that parties do not simulate in the meeting points phase, so $\Delta(B^*) = \Delta(H^*) - \Delta(G^*) \leq -\Delta(G^*)$. The third line follows from the fact that the parties all do meeting points in parallel and the definition of $G^*$: after doing the meeting points phase, there is some pair of parties $(u, v)$ such that the maximum chunk number that they have simulated correctly is the new value of $G^*$. Then, by definition, we have that $\Delta_{u,v}(G^*) = \Delta(G^*)$. The fourth line follows from the fact that parties do not simulate in meeting points, so $\Delta_{u,v}(G^*) \leq 0$.
\end{proof}

The main claim that we establish in this section is that if $status_{u,v}$ or $status_{v,u}$ is ``meeting points'' after the Meeting Points algorithm, then the function $\widetilde{\Delta_{u,v}(\phi)}$ rises by $\Omega(\K)$ in the meeting points phase between parties $u$ and $v$ in the absence of errors and hash collisions. Furthermore, if $status_{u,v} = \text{``simulate''}$ after Meeting Points, then $\varphi_{u,v}$ does not change. Note that this implies that the potential $\phi$ rises by at least $\Omega(c \cdot \K)$ where $c$ is the number of pairs of adjacent parties $(u, v)$ such that $status_{u,v} = \text{``meeting points''}$ and no errors or hash collisions occur between them, by Claim \ref{claim:total-potential-rises}. The proof of this is very similar to the main proof in~\cite{haeupler14}, where he effectively shows that $G_{u,v} - \varphi_{u,v}$ rises by $\Omega(1)$ in each iteration of Meeting Points and Simulation.

In the rest of the section, we will fix parties $u$ and $v$ and focus on how the potential between them changes after they do Meeting Points (Algorithm \ref{alg:meeting-points}) with each other. As noted earlier, we will drop the superscript $\{u,v\}$ off terms such as $k^{\{u,v\}}, E^{\{u,v\}}$, $WM^{\{u,v\}}$, etc. with the understanding that we are only talking about this pairwise interaction.

\begin{proposition}
\label{prop:status-simulate}
Fix parties $u$ and $v$ such that $(u,v) \in E$. If $status_{u,v} = status_{v,u} = \text{``simulate''}$ after Meeting Points, then $\varphi_{u,v}$ is unchanged after Meeting Points, and $\widetilde{\Delta_{u,v}(\phi)} = 0$.
\end{proposition}
\begin{proof}
If we have $status_{u,v} = status_{v,u} = \text{``simulate''}$ after Meeting Points, then all variables in $\varphi_{u,v}$ remain unchanged when the Meeting Points method returns. We note that no kind of truncation occurs in this case, so neither $G_{u,v}$ nor $G^*$ change after the Meeting Points phase in this case, hence $\widetilde{\Delta_{u,v}(\phi)}=0$.
\end{proof}
\begin{lemma}
\label{lem:potential-rises-verification}
Fix parties $u$ and $v$ such that $(u,v) \in E$. Then the \emph{verification} phase of Meeting Points between $u$ and $v$ causes the potential $\varphi_{u,v}$ to rise by at most $5C_6$, and we have that $\widetilde{\Delta_{u,v}(\phi)} \geq -5C_6 \K$ when only taking into account changes from the verification phase. Furthermore, in the absence of errors or hash collisions, $\varphi_{u,v}$ falls by at least five, and $\widetilde{\Delta_{u,v}(\phi)} \geq 5\K$ when only taking into account changes from the verification phase.
\end{lemma}
\begin{proof}
We start by noting that $E$ and $WM$ increment by at most 2 in any verification phase. Furthermore, both $k_{u,v}$ and $k_{v,u}$ are incremented by one: this means that if they start equal, they will stay equal. Conversely, if they start different, they will stay different. Furthermore, note that, since this is a verification phase, no transition occurs, so $B_{i,i+1}$ remains the same. Thus $\varphi_{u,v}$ rises by at most $5C_6$, if we take $C_6$ to be sufficiently large with respect to $C_5$. We now proceed to prove that $\varphi_{u,v}$ falls in the absence of errors and hash collisions. \\
Case $k_{u,v} = k_{v,u}$:

 $WM$ increments only if there was an error or hash collision in the verification phase. Now suppose there is no error or hash collision. In this case, $E$ and $WM$ do not increment, and $k$ increments by 2 (one increment for each party), so $\varphi_{u,v}$ falls by at least five by taking $C_2 > 2.5$. \\
Case $k_{u,v} \neq k_{v,u}$:

Note that $WM$, by definition, increments only in the presence of an error. Suppose that $E$ does not increment by 2. This means that there was a hash collision or error - otherwise, both parties would have incremented $E$. Now suppose that $E$ does increment by 2. In this case, $\varphi_{u,v}$ falls by at least $(-0.9C_4 + C_4)2$. By choosing $C_4$ to be a sufficiently large number, this is at least five.

The conclusions about $\widetilde{\Delta_{u,v}(\phi)}$ follow from the fact that $G_{u,v}$ and $G^*$ remain unchanged in the verification phase, so $\widetilde{\Delta_{u,v}(\phi)} = -\K \varphi_{u,v}$.
\end{proof}

Before proceeding with the proof, we define some notation we will use.

\textbf{Notation for Lemma \ref{lem:potential-rises-overall-1} and its proof (including proofs of helper claims)}:
\begin{itemize}
\item We will often drop superscripts on quantities such as $k^{\{u,v\}}, WM^{\{u,v\}}, E^{\{u,v\}}$, with the understanding that all quantities have $\{u,v\}$ as an implied superscript.
\item Quantities like $k_{u,v}, WM_{u,v}, E_{u,v}$ will refer to the value of these variables just before the transition phase (that is, \emph{after} the verification phase). We note that there is one proposition for which we will use $k_{u,v}$ to denote the value of the variable before the previous verification phase as well - we will be explicit about this abuse of notation in this case.
\item Define $k_{u,v}', WM_{u,v}', E_{u,v}'$ to be the value of the corresponding variables after the transition phase. Define $k_{v,u}', WM_{v,u}', E_{v,u}'$ similarly. Finally, define $k' = k_{u,v}' + k_{v,u}'$, and define $E'$ and $WM'$ similarly.
\item We abuse (our own) previous notation and define $\Delta(k_{u,v})$ to be $k_{u,v}' - k_{u,v}$ - that is, the difference in the variable after transitioning. Define $\Delta(E_{u,v})$, $\Delta(WM_{u,v})$, $\Delta_{u,v}(G^*)$, and $\Delta(B_{u,v})$ similarly as the change incurred in the relevant value by the transition phase between $u$ and $v$.
\end{itemize}

\begin{lemma}
\label{lem:potential-rises-overall-1}
Fix parties $u$ and $v$ such that $(u,v) \in E$, and suppose that $status_{u,v} = \text{``meeting points''}$ after the Meeting-Points phase. Then Meeting-Points phase between $u$ and $v$ causes the potential $\varphi_{u,v}$ to increase by at most $5C_6$. In the absence of errors or hash collisions, $\varphi_{u,v}$ decreases by at least five. Furthermore, we have that $\widetilde{\Delta_{u,v}(\phi)} \geq -5C_6 \cdot \K$, and in the absence of errors or hash collisions, $\widetilde{\Delta_{u,v}(\phi)} \geq 5\K$.
\end{lemma}
Note that the above lemma includes the verifiation and transition phases. We have already established that this is true for the verification phase \emph{alone} in Lemma \ref{lem:potential-rises-verification}. If we could also establish that $\varphi_{u,v}$ does not increase and $\widetilde{\Delta_{u,v}(\phi)}$ does not decrease in the transition phase, we would be able to conclude this lemma as a corollary. However, there is one case for which we cannot do this - in this case, we need to lump together the transition phase with the previous verification phase to argue that the potential rises enough there to pay for a possible decrease in our transition phase. This will occur in Proposition \ref{prop:transition-1}. We now split Lemma \ref{lem:potential-rises-overall-1} into cases and prove each case separately. We assume that there is no error in the transition phase. Since the decisions of the parties here only depend on their state, any error does not affect the transition and its corresponding effect on $\varphi_{u,v}$ or $\widetilde{\Delta_{u,v}(\phi)}$.

\begin{proposition}
\label{prop:transition-1}
Fix parties $u$ and $v$ such that $(u,v) \in E$. Suppose that $k_{u,v} \neq k_{v,u}$, and exactly one party does a meeting point or reset transition. Then the current invocation of Meeting Points \emph{as a whole} causes $\varphi_{u,v}$ to rise by at most $5C_6$ and $\widetilde{\Delta_{u,v}(\phi)} \geq -5C_6 \cdot \K$ . If no error or hash collision occurs, then $\varphi_{u,v}$ falls by at least five, and $\widetilde{\Delta_{u,v}(\phi)} \geq 5\K$.
\end{proposition}
\begin{proof}[Proof of Proposition \ref{prop:transition-1}]

 Suppose wlog that the transitioning party is party $u$. Since party $v$ did \emph{not} transition, we have that $k_{u,v}' \neq k_{v,u}'$ after the transition as well. We now analyze the parameters for party $u$, to analyze how party $u$'s contribution to the potential changes. We assume that $WM_{u,v}$ was initially 0, otherwise the decrease in $WM_{u,v}$ to 0 will increase $\varphi_{u,v}$. Now, we know that party $u$ will set $k_{u,v}$ and $E_{u,v}$ to 0. Setting $k_{u,v}$ to 0 will result in a decrease in potential, and setting $E_{u,v}$ to 0 will result in an increase in potential. The net change of $\varphi_{u,v}$ from these two actions is $-0.9C_4 k_{u,v} + C_4 E_{u,v}$.

Note that for a meeting point transition, we have that $E_{u,v} < 0.5 k_{u,v}$, so the expression above is $-0.9C_4 k_{u,v} + C_4 E_{u,v} \leq -0.4 C_4 k_{u,v}$. Note that a meeting points transition can affect $B_{u,v}$, $G^*$, and $G_{u,v}$ as well. However, that the change in each of these values is at most $2k_{u,v}$, and the constants multiplying them in $\widetilde{\Delta_{u,v}(\phi)}$ are $C_3$, $C_1$, and 1 respectively. Since we can take these to be much smaller than $C_4$, the effect on $\Delta(\varphi_{u,v})$ (resp. $\widetilde{\Delta_{u,v}(\phi)}$) from changes in these variables are negligible compared to $-0.4C_4 k_{u,v}$ (resp. $0.4 C_4 k_{u,v} \cdot \K$). Hence, by taking $C_4$ to be large enough, we get the desired result, that $\varphi_{u,v}$ falls by at least five and $\widetilde{\Delta_{u,v}(\phi)} \geq 5\K$.

Now we turn our attention to reset transitions. Note that for reset transitions, $G_{u,v}$ and $G^*$ are unchanged, so $\widetilde{\Delta_{u,v}(\phi)} = -\K \Delta(\varphi_{u,v})$. Recall that, when $2E_{u,v} \geq k_{u,v}$, $u$ will reset. So, we know that \emph{before} her most recent increment of $k_{u,v}$ in the last verification phase (and possibly $E_{u,v}$), we either had that $k_{u,v}, E_{u,v} = 0$ or $2E_{u,v}$ is strictly less than $k_{u,v}$. This means that, after potentially incrementing $E_{u,v}$ and $k_{u,v}$ in the following iteration, we have that (currently) $E_{u,v} \leq 0.5k_{u,v} + 0.5$.\footnote{Note that this inequality also holds if we started with $k_{u,v}$, $E_{u,v}$ equal to 0.}. Plugging into the expression for potential difference above, we see that
\begin{equation}
\label{eq:transition-1}
\Delta(\varphi_{u,v}) \leq C_4(-0.9 k_{u,v} + E_{u,v}) \leq C_4(-0.4k_{u,v} + 0.5).
\end{equation}
Note that when $k_{u,v} > 1$, this is at most $-5$ for sufficiently large $C_4$, and hence $\widetilde{\Delta_{u,v}(\phi)} \geq 5\K$. Combining this with the rise in potential during the verification phase (Lemma \ref{lem:potential-rises-verification}), we conclude the proof of Proposition \ref{prop:transition-1} in the case when $k_{u,v} > 1$.

When $k_{u,v} = 1$, we note that, before the previous verification phase, we must have had $k_{u,v} = 0$. So if we compare $k_{u,v}$ \emph{before the verification phase} to $k_{u,v}'$ \emph{after the transition phase}, we see that they are both equal. The same is true of $E_{u,v}$ before the verification phase and $E_{u,v}'$. So party $u$'s contribution to $\varphi_{u,v}$ does not increase after one iteration of the meeting points protocol, and all the other terms in $\widetilde{\Delta_{u,v}(\phi)}$ are also unchanged. But what about party $v$? By assumption, party $v$ does not transition. Further, in the absence of errors and hash collisions we know that $status_{v,u}=\text{``meeting points''}$ after Meeting Points, since $k_{u,v} \neq k_{v,u}$. Therefore, $v$ goes through a verification phase. By Lemma \ref{lem:potential-rises-verification} and the fact that $v$ does not transition, we get that the contribution of party $v$ to $\varphi_{u,v}$ decreases by at least five in the absence of errors and hash collisions, and increases by at most $5C_6$ in their presence. Hence, overall, $\varphi_{u,v}$ falls by at least five in the absence of errors, and rises by at most $5C_6$ in the presence of errors.
\end{proof}
\textbf{Note on notation}: $\widetilde{\Delta_{u,v}(\phi)}$ changes meaning for the rest of the proof, to match with $\Delta(\cdot)$.
\smallskip

Proposition \ref{prop:transition-1} was the only case where we needed to lump together verification and transition phases to argue that the potentials behave like we want. For the remainder of the argument, it suffices to show that $\varphi_{u,v}$ falls in the transition phase, and that $\widetilde{\Delta_{u,v}(\phi)} > 0$ rises. Hence, we abuse our previous notation to define $\widetilde{\Delta_{u,v}(\phi)} := (\K/m)\Delta_{u,v}(G_{u,v}) - \K \cdot \Delta(\varphi_{u,v}) + C_1 \K \cdot \Delta_{u,v}(G^*)$, where we recall that $\Delta_{u,v}(\cdot)$ is now defined to be the change in a variable after the transition phase \emph{only}.

\begin{proposition}
\label{prop:transition-2}
Fix parties $u$ and $v$ such that $(u,v) \in E$. Suppose that $k_{u,v} \neq k_{v,u}$, and both parties do some transition. Then $\varphi_{u,v}$ falls by at least one in the transition phase. Furthermore, $\widetilde{\Delta_{u,v}(\phi)} \geq \K$.
\end{proposition}
\begin{proof}[Proof of Proposition \ref{prop:transition-2}]
Since both parties transition, we will have $k_{u,v}' = k_{v,u}' = 0$ after the transition. Furthermore, we will have $E_{u,v}' = E_{v,u}' = WM_{u,v}' = WM_{v,u}' = 0$ due to the transitioning. Hence, $ \Delta(\varphi_{u,v}) \leq -0.9 C_4 k + C_4 E - C_6 WM$. Just like in the proof of Proposition \ref{prop:transition-1}, we note that the reset condition (Line~\ref{step:resetMP}) implies that we have $E_{u,v} \leq 0.5 k_{u,v} + 0.5$, and similarly $E_{v,u} \leq 0.5 k_{v,u} + 0.5$. So $E \leq 0.5 k + 1$, and so we get that
\[
\Delta(\varphi_{u,v}) \leq C_4 (-0.4 k + 1).
\]

Since we know that both $k_{u,v}$ and $k_{v,u}$ are greater than 1 after the verification phase and we also know that $k_{u,v} \neq k_{v,u}$, we conclude that $k_{u,v} + k_{v,u} \geq 3$. Therefore, the above expression is at most $-0.2C_4$. This establishes that $\Delta(\varphi_{u,v}) \leq -1$. To see that indeed $\widetilde{\Delta_{u,v}(\phi)} \geq \K$, we additionally note that $G^*$ and $G_{u,v}$ change by at most $2k$, and so we get that $\widetilde{\Delta_{u,v}(\phi)} \geq \K \cdot (-2k + 0.4 C_4 k - 2C_1 k - C_4) \geq \K \cdot k(0.4C_4 - 2C_1 - 2) - C_4\K$. By taking $C_4$ sufficiently large so that $0.05C_4 \geq 2C_1 + 2$ and using the fact that $k \geq 3$, we see that this in turn is $\geq 0.05 C_4 \K$, which is greater than $\K$ since we needed $C_4 > 20$ earlier.
\end{proof}

\begin{proposition}
\label{prop:transition-3}
Fix parties $u$ and $v$ such that $(u,v) \in E$. Suppose that $k_{u,v} = k_{v,u}$, and exactly one party transitions. Then $\varphi_{u,v}$ falls by at least one in the transition phase. Furthermore, $\widetilde{\Delta_{u,v}(\phi)} \geq \K$.
\end{proposition}
\begin{proof}[Proof of Proposition \ref{prop:transition-3}]
Note that since only one party transitions, $k_{u,v}' \neq k_{v,u}'$. Wlog, suppose that the transitioning party is $u$.

Let us suppose that the transition was a reset transition, so $E_{u,v} \geq 0.5k_{u,v}$. Then $u$'s contribution to $\varphi_{u,v}$ will fall, since it sets $E_{u,v}' = k_{u,v}' = 0$ and we can take $C_5 > 2C_2 + 1$. But a priori it seems possible that party $v$ may have its contribution to $\varphi_{u,v}$ \emph{rise}. This is because the contribution of $k_{v,u}$ to the potential was $-C_2 k_{v,u}$, but after party $u$ transitions the contribution is $0.9C_4 k_{v,u}'$, since after $u$'s transition we are using the potential function for unequal $k$'s.
To address this, we note that $k_{v,u}' = k_{v,u} = k_{u,v} \leq 2E_{u,v}$, where the first equality follows because party $v$ did not transition and the second equality holds by assumption. Hence, the change in $\varphi_{u,v}$ is at most $-C_5 E_{u,v} + C_2k_{u,v} + C_2k_{v,u} + C_4k_{v,u}' \leq (-0.5C_5 + 2C_2 + C_4)k_{u,v}$. For sufficiently large choice of $C_5$, this quantity is at most $-1$. Since $G_{u,v}$ and $G^*$ do not change for a reset trasition, this also implies that $\widetilde{\Delta_{u,v}(\phi)} \geq \K$.

Now let us suppose that the transition was a meeting point transition. For simplicity, assume that party $u$ transitions to meeting point 1; identical reasoning will hold for transitioning to meeting point 2.

Note that $u$ only transitions when $mpc1_{u,v} \geq 0.4k_{u,v}$. Since $v$ is not transitioning, we know that $mpc1_{v,u} < 0.4 k_{v,u} = 0.4 k_{u,v}$ and $mpc2_{v,u} < 0.4 k_{u,v}$. Furthermore, we know that for the last $0.5k_{u,v}$ iterations (if $k_{u,v}=1$, then for 1 iteration), both parties have exchanged hashes of the same meeting points. Either there was truly a match among these meeting points or there was not. If there was truly a match (wlog say it is with $v$'s first meeting point), then we know that $WM \geq WM_{v,u} > 0.1k_{v,u} = 0.05 k$, since $v$ did not increment either $mpc1$ more than $0.4k_{u,v}$ times. Since $k'_{u,v} \neq k'_{v,u}$ after the transition, the $WM$ term in $\varphi_{u,v}$ falls by $0.05C_6 \cdot k$ as a result. If there was not truly a match, then we know that $WM_{u,v} \geq 0.4k_{u,v}$, since $u$ incremented its $mpc1$ counter $0.4k_{u,v}$ times despite the lack of a true match. $WM_{u,v}$ resets to 0 after $u$ transitions, so the of $WM$ term in $\varphi_{u,v}$ falls by at least $0.2C_6 \cdot k$ after the transition.

Furthermore, the contribution of each of the other terms in $\varphi_{u,v}$ is at most $2C_5 \cdot k$ due to this transition. Hence, by taking $C_6$ to be sufficiently large with respect to $C_5$, we get that $\varphi_{u,v}$ falls by at least $\Omega(C_6 \cdot k)$, which is at least 1 for sufficiently large $C_6$. Furthermore, $G_{u,v}$ and $G^*$ also only change by at most $2k$, and so by taking $C_6$ to be sufficiently large, we get that $\widetilde{\Delta_{u,v}(\phi)} \geq \K$.
\end{proof}

\begin{proposition}
\label{prop:transition-4}
Fix parties $u$ and $v$ such that $(u,v) \in E$. Suppose that $k_{u,v} = k_{v,u}$, and both parties transition. Then $\varphi_{u,v}$ falls by at least one in the transition phase. Furthermore, $\widetilde{\Delta_{u,v}(\phi)} \geq \K$.
\end{proposition}
\begin{proof}[Proof of Proposition \ref{prop:transition-4}]
The main difference from Proposition \ref{prop:transition-3} is that we have $k_{u,v}' = k_{v,u}'$ after the transition. Note that $|\Delta(B_{u,v})|, |\Delta(G_{u,v})|$, and $|\Delta_{u,v}(G^*)|$ are all upper bounded by $2k$. Suppose that at least one party (wlog, $u$) does a reset transition. Note that this transition can only decrease $WM$, which in turn only decreases $\varphi_{u,v}$, so we assume that $WM = 0$. Furthermore, recall that $E_{u,v} \geq k_{u,v}/2 = k/4$. Then the difference in $\widetilde{\Delta_{u,v}(\phi)}$ caused after both parties transition is at least
\begin{align*}
\widetilde{\Delta_{u,v}(\phi)} 
&\geq (\Delta(G_{u,v}) + C_5 E_{u,v} + C_2 \Delta(k) - C_3 \Delta(B_{u,v}) 
	\\ & \quad \phantom{(}{}+ C_1 \Delta_{u,v}(G^*))\K \\
&\geq ((C_5/4) k - 2 C_3 k - C_2 k - 2 C_1 k - 2k)\K \\
&\geq ((C_5/4) - 7C_3)k \cdot \K
\end{align*}
By taking $C_5$ to be larger than $28C_3 + 1$ and noting that $k = k_{u,v} + k_{v,u} > 1$, we see that $\widetilde{\Delta_{u,v}(\phi)} > \K$.

Now suppose both parties do a meeting point transition. Suppose that $k_{u,v} > 4 B_{u,v}$. Note that $k_{u,v}$ is a power of two by definition since a meeting point transition is occurring. So if $B_{u,v} > 0$, then $k_{u,v}$ is divisible by 4. Then we must have had at least $k_{u,v}/4$ iterations where $u$ had some value for $k_{u,v}$ that was at least $k_{u,v}/4 > B_{u,v}$, and the same for $v$ with $k_{v,u}$. In this case, one of party $u$'s two meeting points corresponds to $\Pi_{u,v}[1:c(k_{u,v}/4)]$, where $c$ is defined to be the largest integer such that $c \cdot k_{u,v}/4 \leq G_{u,v}$. This uses the fact that $k_{u,v} > 0$ after the Meeting Points phase implies that $status_{u,v} = \text{``meeting points''}$, and that this prevents $u$ from simulating or rewinding $\Pi_{u,v}$. The analogous fact is true of $v$ as well, for identical reasons. Then, by the definition of $G_{u,v}$, we have that
\[  \Pi_{u,v}[1:c(k_{u,v}/4)] =  \Pi_{v,u}[1:c(k_{u,v}/4)].\]
However, the parties did not transition at step $k_{u,v}/2$ -- this means that we must have $WM \geq 0.4(k_{u,v}/2) = 0.1k$. Since $C_6$ is sufficiently large, we get that $\varphi_{u,v}$ falls by at least 1 and that $\widetilde{\Delta_{u,v}(\phi)} \geq \K$.

Note that if $B_{u,v} = 0$, then it is possible that we have $k_{u,v} \in \{1, 2\}$, which we now address for completeness. If $k_{u,v}=1$, then the parties exchanged a single meeting point and are now doing a meeting point transition. But it cannot be that their hashes matched for the first meeting point - if this were the case, they would have made their status ``simulate'' instead of going into the transition phase. But since we know $B_{u,v} = 0$, their hashes should have matched on the first meeting point. Hence, we get that $WM \geq k_{u,v} = 1$. If $k_{u,v}=2$, then note that the parties did not match their Meeting Points when they had previously had $k_{u,v} = k_{v,u} = 1$. Since $B_{u,v}=0$, they should have matched, and so we get that $WM \geq 0.5k_{u,v}$. Hence, either way, $WM$ is sufficiently large so that $\varphi_{u,v}$ falls by at least 1 and that $\widetilde{\Delta_{u,v}(\phi)} \geq \K$.

Now we assume that $k_{u,v}, k_{v,u} \leq 4 B_{u,v}$. Note that this implies that $B_{u,v} > 0$, since $k_{u,v}$ and $k_{v,u}$ are both at least 1. First, we consider the case where $B_{u,v}' \neq 0$. In this case, their communication before now must have had at least $0.4k_{u,v}$ hash collisions or corruptions to make them both increment their meeting point counters enough to transition, and so $WM$ is at least $0.4k_{u,v}$ before the transition. The decrease in $WM$ from the transition means that $\varphi_{u,v}$ falls by at least 1 and that $\widetilde{\Delta_{u,v}(\phi)} \geq \K$.

Now assume that $B_{u,v}' = 0$. Then the potential change from this meeting point transition is at least
\begin{align*} 
&\K \cdot \left(\Delta(G_{u,v}) + C_1\Delta(G^*) + C_2 \Delta(k) - C_3 \Delta(B_{u,v}) \right) \\
&\geq \K \cdot \left( (1+C_1+C_2) (-4k) - C_3 (-k/8) \right).
\end{align*}
 By taking $C_3$ to be large enough, we get that $\widetilde{\Delta_{u,v}(\phi)} \geq \K$.

Since $\widetilde{\Delta_{u,v}(\phi)} > \K$, $\varphi_{u,v}$ falls by at least one, as $\Delta_{u,v}(G^*)$ and $\Delta(G_{u,v})$ are both nonpositive.
\end{proof}
\begin{proof}[Proof of Lemma \ref{lem:potential-rises-overall-1}]
First, note that an argument basically identical to the proof of Proposition \ref{prop:status-simulate} shows that any party $v$ with $status_{v,u} = \text{``simulate''}$ after Meeting Points does not contribute any change to $\varphi_{u,v}$. Hence, we can ignore these parties when establishing Lemma \ref{lem:potential-rises-overall-1}.
\begin{itemize}
\item Suppose neither party transitions. Then the only change to $\varphi_{u,v}$ is in the verification phase, and Lemma \ref{lem:potential-rises-verification} establishes the claim.
\item Suppose one party transitions and $k_{u,v} \neq k_{v,u}$. Then Proposition \ref{prop:transition-1} establishes the claim.
\item Suppose both parties transition and $k_{u,v} \neq k_{v,u}$. Then Lemma \ref{lem:potential-rises-verification} and Proposition \ref{prop:transition-2} establish the claim.
\item Suppose one party transitions and $k_{u,v} = k_{v,u}$. Then Lemma \ref{lem:potential-rises-verification} and Proposition \ref{prop:transition-3} establish the claim.
\item Suppose both parties transition and $k_{u,v} = k_{v,u}$. Then Lemma \ref{lem:potential-rises-verification} and Proposition \ref{prop:transition-4} establish the claim.
\end{itemize}
\end{proof}

Now we can prove the final lemma of this section, which was used directly in the proof that the potential $\phi$ rises during the Meeting Points phase (Lemma~\ref{lem:phi-rises-enough}).

Note that, now that we are done with the proof of Lemma \ref{lem:potential-rises-overall-1}, the meaning of $\Delta(\cdot)$ goes back to the original one; namely, it gives the total change in some variable over a single invocation of the Meeting Points protocol, instead of only reflecting the change from the transition phase.
\begin{lemma}
\label{lem:network-potential-rises}
Let $c$ be the number of pairs $(u,v) \in E$ such that $status_{u,v}$ or $status_{v,u}$ is ``meeting points'' after the Meeting Points phase. Let $\ell_1$ denote the total number of links with errors and hash collisions during the Meeting Points phase. Then after all adjacent parties do Meeting Points, the overall potential rise $\Delta(\phi)$ is at least $5c \cdot \K + 0.4C_7 \ell_1 \cdot \K$.
\end{lemma}
\begin{proof}[Proof of Lemma \ref{lem:network-potential-rises}]
By Claim \ref{claim:total-potential-rises}, we recall that
\[\Delta(\phi) \geq \sum_{(u,v) \in E} \left(\widetilde{\Delta_{u,v}(\phi)} + C_7 \K \Delta_{u,v}(EHC)\right) \]

For any pair $(u,v)$ that have an error or hash collision between them during the Meeting Points phase, $\Delta_{u,v}(EHC) \geq 1$, so we get that $\widetilde{\Delta_{u,v}(\phi)} + C_7 \K \Delta_{u,v}(EHC) \geq C_7 \cdot \K - 5C_6 \cdot \K \geq 0.5 C_7 \cdot \K$, where this follows from Lemma \ref{lem:potential-rises-overall-1} in the case where either $status_{u,v}$ or $status_{v,u}$ was ``meeting points''. In the case where both parties have $status_{u,v} = \text{``simulate,''}$, Proposition \ref{prop:status-simulate} gives us that $\varphi_{u,v}$ is unchanged after Meeting Points, and so are $G^*$ and $G_{u,v}$. Hence, $\widetilde{\Delta_{u,v}(\phi)} + C_7 \K \Delta_{u,v}(EHC) \geq C_7 \cdot \K$.

For any pair $(u,v)$ such that $status_{u,v} = \text{`` meeting points''}$ and there is no error between them, Lemma \ref{lem:potential-rises-overall-1} gives us that $\widetilde{\Delta_{u,v}(\phi)} \geq 5\K$.

For any pair $(u,v)$ such that $status_{u,v} = status_{v,u} = \text{``simulate''}$ after Meeting Points and there is no error between them, Proposition \ref{prop:status-simulate} gives us that $\widetilde{\Delta_{u,v}(\phi)} \geq 0$.

Let $s$ denote the number of links $(u,v)$ for which $status_{u,v}$ or $status_{v,u}$ is ``meeting points'' and there is an error or hash collision on the link. 
We have that
\begin{align*}
\Delta(\phi) &\geq \sum_{(u,v) \in E} \left( \widetilde{\Delta_{u,v}(\phi)} + C_7 \K \Delta_{u,v}(EHC)\right) \\
&\geq 5(c-s) \cdot \K + 0.5 C_7 \ell_1 \cdot \K \\
&\geq 5(c-s) \cdot \K + (0.4 C_7 + 5) \ell_1 \cdot \K \\
&\geq 5c \cdot \K + 0.4 C_7 \ell_1 \cdot \K
\end{align*}
The second line follows from the previous analyses in the lemma. The third line follows from taking $C_7\geq 50$. The final line follows from noting that $s \leq \ell_1$ by definition, since it counts a subset of the links with an error or hash collision. This completes the proof.
\end{proof}

\end{document}